\documentclass{sigplanconf}
\usepackage{proof}
\usepackage{amssymb}
\usepackage{amsmath}
\usepackage{txfonts}
\newtheorem{lemma}{Lemma}[section]

\def\lshfit#1#2{\kern-#1 #2\kern#1}

\def\fillsquare{\kern2pt\raise0.25pt
     \hbox{$\vcenter{\hrule height0pt \hbox{\vrule width5pt height5pt} \hrule height0pt}$}}
\newenvironment{proof}{%
\vskip6pt{\em Proof\kern6pt}}{%
{\unskip\nobreak\hfill\penalty50\kern4pt\hbox{}\nobreak\hfill\fillsquare}\vskip6pt}
\begin{document}

\setlength{\pdfpageheight}{\paperheight}
\setlength{\pdfpagewidth}{\paperwidth}

\conferenceinfo{}{}
\copyrightyear{2015} 
\copyrightdata{978-1-nnnn-nnnn-n/yy/mm} 
\doi{nnnnnnn.nnnnnnn}

\title{%
Session Types in a Linearly Typed Multi-Threaded\break Lambda-Calculus
}

\authorinfo%
{Hongwei Xi\and Zhiqiang Ren\and Hanwen Wu\and William Blair}{Boston University}{\{hwxi,aren,hwwu,wdblair\}@cs.bu.edu}
\maketitle

\def\MTLC{\mbox{MTLC}}

\begin%
{abstract}

We present a formalization of session types in a multi-threaded
lambda-calculus (MTLC) equipped with a linear type system,
establishing for the MTLC both type preservation and global progress.
The latter (global progress) implies that the evaluation of a
well-typed program in the MTLC can never reach a deadlock.  As this
formulated MTLC can be readily embedded into ATS, a full-fledged
language with a functional programming core that supports both
dependent types (of DML-style) and linear types, we obtain a direct
implementation of session types in ATS. In addition, we gain immediate
support for a form of dependent session types based on this embedding
into ATS. Compared to various existing formalizations of session
types, we see the one given in this paper is unique in its closeness
to concrete implementation. In particular, we report such an
implementation ready for practical use that generates Erlang code from
well-typed ATS source (making use of session types), thus taking great
advantage of the infrastructural support for distributed computing in
Erlang.

\end{abstract}
\def\simp{\rightarrow}
\def\timp{\rightarrow}
\def\ctimp{\Rightarrow}
\def\itimp{\rightarrow_{i}}
\def\ltimp{\rightarrow_{l}}
\def\cres{rc}
\def\const{c}
\def\cfun{\mbox{\it cf}}
\def\ccon{\mbox{\it cc}}
\def\ctrue{\mbox{\it true}}
\def\cfalse{\mbox{\it false}}
\def\exp{e}
\def\vexp{\vec{\exp}}
\def\val{v}
\def\vval{\vec{\val}}
\def\xf{\mbox{\it x{\kern0.5pt}f}}
\def\dif{\mbox{\tt if}}
\def\dfst#1{\mbox{\tt fst}(#1)}
\def\dsnd#1{\mbox{\tt snd}(#1)}
\def\dunit{\langle\rangle}
\def\tuple#1{\langle#1\rangle}
\def\lam#1#2{{\tt lam}\;#1.\,#2}
\def\app#1#2{{\tt app}(#1, #2)}
\def\fix#1#2{{\tt fix}\;#1.\,#2}
\def\letin#1#2{{\tt let}\;#1\;{\tt in}\;#2\;{\tt end}}
\def\ty{T}
\def\vw{V}
\def\vwty{\hat{T}}
\def\tvar{\alpha}
\def\vtvar{\hat{\alpha}}
\def\tunit{\mbox{\bf 1}}
\def\tbase{\delta}
\def\tint{\mbox{\bf int}}
\def\tbool{\mbox{\bf bool}}
\def\vtbase{\hat{\tbase}}
\def\tpjg{\vdash}
\def\temd{\models}
\def\SIG{\mbox{\rm SIG}}
\def\langz{\MTLC_0}
\def\langch{\MTLC_{{\rm ch}}}
\def\ST{S}
\def\stvar{\sigma}
\def\stnil{\mbox{\tt nil}}
\def\stnild{\overline{\stnil}}
\def\dual#1{\mbox{\it dual}(#1)}
\def\chsnd#1#2{\mbox{\tt snd}(#1)::#2}
\def\chrcv#1#2{\mbox{\tt rcv}(#1)::#2}
\def\tchpos{\mbox{\bf chpos}}
\def\tchneg{\mbox{\bf chneg}}
\def\fsendpos{\mbox{\it send}}
\def\fsendneg{\underline{\mbox{\it send}}}
\def\frecvpos{\mbox{\it recv}}
\def\frecvneg{\underline{\mbox{\it recv}}}
\def\fclosepos{\mbox{\it close}}
\def\fcloseneg{\underline{\mbox{\it close}}}
\def\fchnegcreate{\mbox{\it chneg\_create}}
\def\fchposneglink{\mbox{\it chposneg\_link}}
\def\chpcst{\mbox{\it ch}}
\def\chncst{\overline{\mbox{\it ch}}}
\def\fchnegcreatetwo{\mbox{\it chneg\_create2}}
\section{Introduction}
\label{section:introduction}
In broad terms, a (dyadic) session is an interaction between two
concurrently running programs, and a session type is a form of type for
specifying (or classifying) sessions. As an example, let us assume that two
programs P and Q are connected with a bidirectional channel. From the
perspective of P, the channel may be specified by a term sequence of the
following form:
$$\chsnd{\tint}{\chsnd{\tint}{\chrcv{\tbool}{\stnil}}}$$
which means that an integer is to be sent, another integer is to be
sent, a boolean is to be received, and finally the channel is to be closed.
Clearly, from the perspective of Q, the channel should be specified
by the following term sequence:
$$\chrcv{\tint}{\chrcv{\tint}{\chsnd{\tbool}{\stnil}}}$$ which means
precisely the dual of what the previous term sequence does.  We may think
of P as a client who sends two integers to the server Q and then receives
from Q either true or false depending on whether or not the first sent
integer is less than the second one.

The session between P and Q is bounded in the sense that it contains only a
bounded number of sends and receives. By introducing recursively defined
session types, we can specify unbounded sessions containing indefinite
numbers of sends and receives.

\def\tchan{\mbox{\bf chan}}
\begin%
{figure}
\fontsize{8pt}{9pt}\selectfont
\begin%
{verbatim}
fun P() = let
  val () = channel_send (CH, I1)
  val () = channel_send (CH, I2)
  val b0 = channel_recv (CH)
  val () = channel_close (CH)
in b0 end (* end of [P] *)

fun Q() = let
  val i1 = channel_recv (CH)
  val i2 = channel_recv (CH)
  val () = channel_send (CH, i1 < i2)
  val () = channel_close (CH)
in () end (* end of [Q] *)
\end{verbatim}
\caption{Some pseudo code in ML-like syntax}
\label{figure:P-and-Q}
\end{figure}
In Figure~\ref{figure:P-and-Q}, we present some pseudo code showing a
plausible way to implement the programs P and Q. Please note that the
functions P and Q, though written together here, can be written in separate
contexts. We use $\mbox{CH}$ to refer to a channel available in the
surrounding context of the code and $\mbox{I1}$ and $\mbox{I2}$ for two
integers; the functions $\mbox{\tt channel\_send}$ and $\mbox{\tt
  channel\_recv}$ are for sending and receiving data via a given channel,
and $\mbox{\tt channel\_close}$ for closing a given channel.  Let us now
sketch a way to make the above pseudo code typecheck.  We can assign the
following type to $\mbox{\tt channel\_send}$:
$$(!\tchan(\chsnd{\ty}{\ST}) \gg \tchan(\ST), \ty) \timp \tunit$$ where
$\ty$ stands for a type and $\ST$ for a session type. Basically, this type
means that calling $\mbox{\tt channel\_send}$ on a channel of the type
$\tchan(\chsnd{\ty}{\ST})$ and a value of the type $\ty$ returns a unit
while {\it changing} the type of the channel to $\tchan(\ST)$.  Clearly,
$\tchan$ must be a linear type constructor for this to make sense. As can
be expected, the type assigned to $\mbox{\tt channel\_recv}$ should be of
the following form:
$$(!\tchan(\chrcv{\ty}{\ST}) \gg \tchan(\ST)) \timp \ty$$ which means that
calling $\mbox{\tt channel\_recv}$ on a channel of the type
$\tchan(\chrcv{\ty}{\ST})$ returns a value of the type $\ty$ while changing
the type of the channel to $\tchan(\ST)$.  As for $\mbox{\tt
  channel\_close}$, it is assigned the following
type: $$(\tchan(\stnil))\timp\tunit$$ indicating that calling $\mbox{\tt
  channel\_close}$ on a channel consumes the channel (so that the channel
is no longer available for use).

ATS~\cite{ATS-types03,CPwTP} is a full-fledged language with a functional
programming core based on ML that supports both dependent types (of
DML-style~\cite{XP99,DML-jfp07}) and linear types.  Its highly expressive
type system makes it largely straightforward to implement session types in
ATS (e.g., based on the outline given above) if our concern is primarily
about type-correctness, that is, finding a way to assign proper types to
various primitive (or built-in) session-type functions such as $\mbox{\tt
  channel\_send}$ and $\mbox{\tt channel\_recv}$ so that type-errors can be
issued if communication protocols specified by session types are not
correctly followed. For instance, there have already been implementations
of session types in Haskell (e.g.,~\cite{NeubauerT04,PucellaT08}) and
elsewhere that offer type-correctness. However, mere type-correctness is
clearly inadequate. We are to present a concrete example showing that
deadlocking can be easily introduced if certain primitive session-typed
functions are supported inadvertently. We want to go beyond
type-correctness. In particular, we are interested in proving formally the
property that concurrency based on session types (as those implemented in
ATS) can never result in deadlocking, which is often referred to as {\em
  global progress}. So we need to formalize session types. Furthermore, we
expect to have a formalization for session types that can greatly
facilitate the determination of deadlock-freeness of certain specific
primitive session-typed functions.

There have been many formalizations of session types in the literature
(e.g.,
~\cite{Honda93,HondaVK98,CastagnaDGP09,GayV10,Vasconcelos12,Wadler12,ToninhoCP13}).
Often the dynamics formulated in a formalization of session types is
based on $\pi$-calculus~\cite{PiCalculus} or its variants/likes. In
our attempt to implement session types in ATS~\cite{ats-lang}, we
found that a formalization of session types based on multi-threaded
$\lambda$-calculus (MTLC) can be of great value due to its closeness
to the underlying implementation language (that is, ATS in our
case). Such a formalization~\cite{LTSMP} is less abstract and more
operational and is also amenable to extension. For instance,
multi-party sessions~\cite{HondaYC08} can be directly introduced into
a MTLC-based formalization of session types. On the other hand,
supporting multi-party sessions in a logic-based formalization of
session-types is yet a great challenge.

Before moving on with formal development, we would now like to use the
moment to slightly modify the example presented above so as to make it
easier for the reader to access the formalization of session types to be
presented.  As is pointed out above, each session type has its own dual and
the dual of its dual equals itself. However, we do not plan to make use of
the notion of dual of a session type explicitly in this paper. Instead, we
are to introduce two kinds of channels: positive channels and negative
channels.  The type for a positive channel specified by a session type is
considered the dual of the type for a negative channel specified by the
same session type and vice versa. More formally, we use $\tchpos(\ST)$ and
$\tchneg(\ST)$ for a positive channel and a negative channel specified by
$\ST$, respectively. One may think of $\tchpos(\ST)$ and
$\tchneg(\dual{\ST})$ being equal and $\tchpos(\dual{\ST})$ and
$\tchneg(\ST)$ being equal, where $\dual{\ST}$ refers to the dual of $\ST$.
The function $\mbox{\tt channel\_send}$ splits into a positive version
$\mbox{\tt chanpos\_send}$ of the following type:
$$(!\tchpos(\chsnd{\vwty}{\ST}) \gg \tchpos(\ST), \vwty) \timp \tunit$$
and a negative version $\mbox{\tt channeg\_recv}$ of the following type:
$$(!\tchneg(\chrcv{\vwty}{\ST}) \gg \tchneg(\ST), \vwty) \timp \tunit$$
where $\vwty$ ranges over linear types (which include non-linear types as a
special case). In the actual formalization, we use the following
equivalent type for \mbox{\tt chanpos\_send}:
$$(\tchpos(\chsnd{\vwty}{\ST}), \vwty) \timp \tchpos(\ST)$$
and a similar one for \mbox{\tt channeg\_recv} so as to simplify
the presentation. Similarly, the function $\mbox{\tt channel\_recv}$ splits into
a positive version $\mbox{\tt chanpos\_recv}$ of the following type
$$(!\tchpos(\chrcv{\vwty}{\ST}) \gg \tchpos(\ST)) \timp \vwty$$
and a negative version $\mbox{\tt channeg\_send}$ of the following type:
$$(!\tchneg(\chsnd{\vwty}{\ST}) \gg \tchneg(\ST)) \timp \vwty$$
In the actual formalization, we use the following
equivalent type for \mbox{\tt chanpos\_recv}:
$$(\tchpos(\chrcv{\vwty}{\ST})) \timp \tchpos(\ST)\otimes\vwty$$ and a
similar one for \mbox{\tt channeg\_recv}, where $\otimes$ forms a
linear tuple type.  For $\mbox{\tt channel\_close}$, the positive
version and negative version are named $\mbox{\tt chanpos\_close}$ and
$\mbox{\tt channeg\_close}$, respectively.\footnote{Overloading in ATS
  can allow the programmer to still use the unsplit function names and
  thus avoid clutter in coding.}  When reading the pseudo code in
Figure~\ref{figure:P-and-Q}, please note that the $\mbox{CH}$ in the
body of P refers to a negative channel (due to the assumption that P
is a client) and the $\mbox{CH}$ in the body of Q a positive channel.

The rest of the paper is organized as follows.  In
Section~\ref{section:langz}, we formulate a multi-threaded
$\lambda$-calculus $\langz$ equipped with a simple linear type system,
setting up the basic machinery for further development. We then extend
$\langz$ to $\langch$ in Section~\ref{section:langch} with support for
session types and establish both type preservation and global progress for
$\langch$. We give interpretation to some linear logic connectives in
Section~\ref{section:LLinterp} to facilitate understanding of session
types, and briefly mention some issues on implementing session types. We
next present a couple of interesting examples in
Section~\ref{section:examples} to illustrate programming with session
types.  Lastly, we discuss some closely related work in
Section~\ref{section:related-work-conclusion} and then conclude.

The primary contribution of the paper consists of a novel formalization of
session-types. Compared to various existing formalizations of session
types(e.g.,
~\cite{Honda93,HondaVK98,CastagnaDGP09,GayV10,Vasconcelos12,Wadler12,ToninhoCP13}), we
see this one being unique in its closeness to concrete
implementation. Indeed, we report an implementation of session types
ready for practical use that generates Erlang code from well-typed ATS
source. The primary technical contribution of the paper lies in a simple and
general approach to showing that concurrency based on session types is
deadlock-free.  For this, a novel notion of DF-reducibility (where DF stands
for deadlock-freeness) is introduced.

\section%
{$\langz$ with Linear Types}
\label{section:langz}
\begin%
{figure}
\[%
\fontsize{8pt}{9pt}\selectfont
\begin%
{array}{lrcl}
\mbox{expr.} & \exp & ::= &
x \mid f \mid \cres \mid \const(\vexp) \mid \\
             &      &     &
\dunit \mid \tuple{\exp_1,\exp_2} \mid \dfst{\exp} \mid \dsnd{\exp} \mid \\
             &      &     &
\letin{\tuple{x_1,x_2}=\exp_1}{\exp_2} \mid \\
             &      &     &
\lam{x}{\exp} \mid \app{\exp_1}{\exp_2} \mid \fix{f}{\val}  \\
\mbox{values} & v & ::= &
x \mid \cres \mid \ccon(\vval) \mid \dunit \mid \tuple{\val_1,\val_2} \mid \lam{x}{\exp} \\
\mbox{types} & \ty & ::= &
\tvar \mid \tbase \mid \tunit \mid \ty_1*\ty_2 \mid \vwty_1\itimp\vwty_2  \\
\mbox{viewtypes} & \vwty & ::= & \vtvar \mid \vtbase \mid \ty \mid \vwty_1\otimes\vwty_2 \mid \vwty_1\ltimp\vwty_2 \\
\mbox{int. expr. ctx.\kern-6pt} & \Gamma & ::= & \emptyset \mid \Gamma, \xf:\ty \\
\mbox{lin. expr. ctx.\kern-6pt} & \Delta & ::= & \emptyset \mid \Delta, x:\vwty \\
\end{array}\]
\caption{Some syntax for $\langz$}
\label{figure:langz:syntax}
\end{figure}
We first present a multi-threaded lambda-calculus $\langz$ equipped with a
simple linear type system, setting up the basic machinery for further
development. The dynamic semantics of $\langz$ can essentially be seen as
an abstract form of evaluation of multi-threaded programs.

Some syntax of $\langz$ is given in Figure~\ref{figure:langz:syntax}.
We use $x$ for a lam-variable and $f$ for a fix-variable, and $\xf$ for
either a lam-variable or a fix-variable. Note that a lam-variable is
considered a value but a fix-variable is not.
We use $\cres$ for constant resources and $\const$ for constants, which
include both constant functions $\cfun$ and constant constructors $\ccon$.
We treat resources in $\langz$ abstractly and will later introduce
communication channels as a concrete form of resources.
The meaning of various standard forms of expressions in $\langz$
should be intuitively clear.  We may refer to a closed expression (that is,
an expression containing no free variables) as a {\em program}.

We use $\ty$ and $\vwty$ for (non-linear) types and (linear)
viewtypes, respectively, and refer $\vwty$ to as a true viewtype if it
is a viewtype but not a type. We use $\tbase$ and $\vtbase$ for base
types and base viewtypes, respectively.  For instance, $\tbool$ is the
base type for booleans and $\tint$ for integers.  For a simplified
presentation, we do not introduce any concrete base viewtypes in
$\langz$. We assume a signature $\SIG$ for assigning a viewtype to
each constant resource $\cres$ and a constant type (c-type) of the
form $(\vwty_1,\ldots,\vwty_n)\ctimp\vwty$ to each constant.  We use
$\tvar$ and $\vtvar$ for variables ranging over types and viewtypes,
respectively, but we do not support explicit quantification over these
variables until Section~\ref{section:additional}.

Note that a type is always considered a viewtype.  Let $\vwty_1$ and
$\vwty_2$ be two viewtypes. The type constructor $\otimes$ is based on
multiplicative conjunction in linear logic.  Intuitively, if a
resource is assigned the viewtype $\vwty_1\otimes\vwty_2$, then the
resource is a conjunction of two resources of viewtypes $\vwty_1$ and
$\vwty_2$.  The type constructor $\ltimp$ is essentially based on
linear implication $\multimap$ in linear logic.  Given a function of
the viewtype $\vwty_1\ltimp\vwty_2$ and a value of the viewtype
$\vwty_1$, applying the function to the value yields a result of the
viewtype $\vwty_2$ while the function itself is consumed. If the
function is of the type $\vwty_1\itimp\vwty_2$, then applying the
function does not consume it. The subscript $i$ in $\itimp$ is often
dropped, that is, $\timp$ is assumed to be $\itimp$ by default.  The
meaning of various forms of types and viewtypes is to be made clear
and precise when the rules are presented for assigning viewtypes to
expressions in $\langz$.

\def\tchoose{\mbox{\bf choose}} There is no type constructor in
$\langz$ based on additive disjunction in linear logic denoted by
$\oplus$ (but such a type constructor is fully supported in ATS), and
this omission is entirely for the sake of a simplified
presentation. There are also multiplicative disjunction ($\invamp$)
and additive conjunction ($\&$) in linear logic~\cite{GIRARD87}.  If
we see viewtypes negatively in the sense that they are for classifying
capabilities (spaces) of consuming (storing) resources, then
$\vwty_1\invamp\vwty_2$ essentially means the capability (space) that
joins two classified by $\vwty_1$ and $\vwty_2$.  We can interpret
$\vwty_1\&\vwty_2$ as a choice to obtain any capability (space) that
can be classified by either $\vwty_1$ or $\vwty_2$. There is no type
constructor corresponding to $\invamp$ in ATS.  As for
$\vwty_1\&\vwty_2$, we can use the following dependent type in ATS to
replace it:
$$\forall b:bool.~\tbool(b)\ltimp\tchoose(\vwty_1, \vwty_2, b)$$ where
$\tbool(b)$ is a singleton type for the only boolean value equal to $b$ and
$\tchoose(\vwty_1, \vwty_2, b)$ equals $\vwty_1$ or $\vwty_2$ depending
whether $b$ is true or false, respectively.

\def\fthreadcreate{\mbox{\it thread\_create}}
There is a special constant function $\fthreadcreate$ in $\langz$
for thread creation, which is assigned the following rather interesting
c-type:
\[\begin{array}{rcl}
\fthreadcreate & : & (\tunit\ltimp\tunit) \ctimp \tunit
\end{array}\]
A function of the type $\tunit\ltimp\tunit$ is a procedure that takes no
arguments and returns no result (when its evaluation terminates).  Given
that $\tunit\ltimp\tunit$ is a true viewtype, a procedure of this type may
contain resources and thus must be called exactly once. The operational
semantics of $\fthreadcreate$ is to be formally defined later.

\def\dom{\mbox{\bf dom}}
\def\mapadd#1#2#3{#1[#2\mapsto #3]}
\def\mapdel#1#2{#1\backslash#2}
\def\maprep#1#2#3{#1[#2 := #3]}
A variety of mappings, finite or infinite, are to be introduced in the
rest of the presentation.  We use $[]$ for the empty mapping and
$[i_1,\ldots,i_n\mapsto o_1,\ldots,o_n]$ for the finite mapping that maps
$i_k$ to $o_k$ for $1\leq k\leq n$.  Given a mapping $m$, we write
$\dom(m)$ for the domain of $m$. If $i\not\in\dom(m)$, we use
$\mapadd{m}{i}{o}$ for the mapping that extends $m$ with a link from $i$ to
$o$.  If $i\in\dom(m)$, we use $\mapdel{m}{i}$ for the mapping obtained
from removing the link from $i$ to $m(i)$ in $m$, and $\maprep{m}{i}{o}$
for $\mapadd{(\mapdel{m}{i})}{i}{o}$, that is, the mapping obtained from
replacing the link from $i$ to $m(i)$ in $m$ with another link from $i$ to
$o$.

\def\resof{\rho}
\begin{figure}
\[\begin{array}{rcl}
\resof(\cres) & = & \{\cres\} \\
\resof(\const(\exp_1,\ldots,\exp_n)) & = & \resof(\exp_1)\uplus\cdots\uplus\resof(\exp_n) \\
\resof(\xf) & = & \emptyset \\
\resof(\dunit) & = & \emptyset \\
\resof(\tuple{\exp_1,\exp_2}) & = & \resof(\exp_1)\uplus\resof(\exp_2) \\
\resof(\dfst{\exp}) & = & \resof(\exp) \\
\resof(\dsnd{\exp}) & = & \resof(\exp) \\
\resof(\dif(\exp_0,\exp_1,\exp_2)) & = & \resof(\exp_0)\uplus\resof(\exp_1) \\
\resof(\letin{\tuple{x_1,x_2}=\exp_1}{\exp_2}) & = & \resof(\exp_1)\uplus\resof(\exp_2) \\
\resof(\lam{x}{\exp}) & = & \resof(\exp) \\
\resof(\app{\exp_1}{\exp_2}) & = & \resof(\exp_1)\uplus\resof(\exp_2) \\
\resof(\fix{f}{\val}) & = & \resof(\val) \\
\end{array}\]
\caption{The definition of $\resof(\cdot)$}
\label{figure:resof}
\end{figure}
We define a function $\resof(\cdot)$ in Figure~\ref{figure:resof} to
compute the multiset (that is, bag) of constant resources in a given
expression.  Note that $\uplus$ denotes the multiset union. In the type
system of $\langz$, it is to be guaranteed that $\resof(\exp_1)$ equals
$\resof(\exp_2)$ whenever an expression of the form
$\dif(\exp_0,\exp_1,\exp_2)$ is constructed, and this justifies
$\resof(\dif(\exp_0,\exp_1,\exp_2))$ being defined as
$\resof(\exp_0)\uplus\resof(\exp_1)$.

\def\rns{R}
\def\RES{{\bf RES}}
We use $\rns$ to range over finite multisets of resources. Therefore,
$\rns$ can also be regarded as a mapping from resources to natural numbers:
$\rns(\cres) = n$ means that there are $n$ occurrences of $\cres$ in
$\rns$. It is clear that we may not combine resources arbitrarily. For
instance, we may want to exclude the combination of one resource stating
integer 0 at a location L and another one stating integer 1 at the same
location. We fix an abstract collection $\RES$ of finite multisets of
resources and assume the following:
\begin{itemize}
\item
$\emptyset\in\RES$.
\item
For any $\rns_1$ and
$\rns_2$, $\rns_2\in\RES$ if $\rns_1\in\RES$ and $\rns_2\subseteq\rns_1$,
where $\subseteq$ is the subset relation on multisets.
\end{itemize}
We say that $\rns$ is a valid multiset of resources if $\rns\in\RES$ holds.

\def\pool{\Pi}
\def\tid{\mbox{\it tid}}
\def\esubst{\theta}
\def\tsubst{\Theta}

In order to formalize threads, we introduce a notion of {\em
  pools}. Conceptually, a pool is just a collection of programs (that is,
closed expressions).  We use $\pool$ for pools, which are formally defined
as finite mappings from thread ids (represented as natural numbers) to
(closed) expressions in $\langz$ such that $0$ is always in the domain of
such mappings.  Given a pool $\pool$ and $\tid\in\dom(\pool)$, we refer to
$\pool(\tid)$ as a thread in $\pool$ whose id equals $\tid$. In particular,
we refer to $\pool(0)$ as the main thread in $\pool$.  The definition of
$\resof(\cdot)$ is extended as follows to compute the multiset of resources
in a given pool:
\begin%
{center}
$\resof(\pool)=\biguplus_{\tid\in\dom(\pool)}\resof(\pool(\tid))$
\end{center}
We are to define a relation on pools in
Section~\ref{section:langz:dynamics} to simulate multi-threaded program
execution.

\begin{figure}
\fontsize{8}{8}\selectfont
\[\begin{array}{c}
\infer[\mbox{\bf (ty-res)}]
      {\Gamma;\emptyset\tpjg\cres:\vtbase}
      {\SIG\temd\cres:\vtbase} \\[6pt]
\infer[\mbox{\bf (ty-cst)}]
      {\Gamma;\Delta_1,\ldots,\Delta_n\tpjg\const(\exp_1,\ldots,\exp_n):\vwty}
      {$$\begin{array}{c}
       \SIG\temd\const:(\vwty_1,\ldots,\vwty_n)\ctimp\vwty \\
       \Gamma;\Delta_i\tpjg\exp_i:\vwty_i~~\mbox{for $1\leq i\leq n$} \\
       \end{array}$$} \\[6pt]
\infer[\mbox{\bf (ty-var-i)}]
      {(\Gamma,\xf:\ty;\emptyset)\tpjg \xf:\ty}
      {} \\[6pt]
\infer[\mbox{\bf (ty-var-l)}]
      {(\Gamma;\emptyset,x:\vwty)\tpjg x:\vwty}
      {} \\[6pt]
\infer[\mbox{\bf(ty-if)}]
      {\Gamma;\Delta_0,\Delta\tpjg\dif(\exp_0,\exp_1,\exp_2):\vwty}
      {$$\begin{array}{c}
       \Gamma;\Delta_0\tpjg\exp_0:\tbool \\
       \Gamma;\Delta\tpjg \exp_1:\vwty \kern6pt
       \Gamma;\Delta\tpjg \exp_2:\vwty \kern6pt
       \resof(\exp_1) = \resof(\exp_2) \\
       \end{array}$$} \\[6pt]
\infer[\mbox{\bf(ty-unit)}]
      {\Gamma;\emptyset\tpjg\dunit:\tunit}{} \\[6pt]
\infer[\mbox{\bf(ty-tup-i)}]
      {\Gamma;\Delta_1,\Delta_2\tpjg\tuple{\exp_1,\exp_2}:\ty_1*\ty_2}
      {\Gamma;\Delta_1\tpjg\exp_1:\ty_1 &
       \Gamma;\Delta_2\tpjg\exp_2:\ty_2 } \\[6pt]
\infer[\mbox{\bf(ty-fst)}]
      {\Gamma;\Delta\tpjg\dfst{\exp}:\ty_1}
      {\Gamma;\Delta\tpjg\exp:\ty_1*\ty_2}
\kern12pt
\infer[\mbox{\bf(ty-snd)}]
      {\Gamma;\Delta\tpjg\dsnd{\exp}:\ty_2}
      {\Gamma;\Delta\tpjg\exp:\ty_1*\ty_2} \\[6pt]
\infer[\mbox{\bf(ty-tup-l)}]
      {\Gamma;\Delta_1,\Delta_2\tpjg\tuple{\exp_1,\exp_2}:\vwty_1\otimes\vwty_2}
      {\Gamma;\Delta_1\tpjg\exp_1:\vwty_1 &
       \Gamma;\Delta_2\tpjg\exp_2:\vwty_2 } \\[6pt]
\infer[\mbox{\bf(ty-tup-l-elim)}]
      {\Gamma;\Delta_1,\Delta_2\tpjg\letin{\tuple{x_1,x_2}=\exp_1}{\exp_2}:\vwty}
      {$$\begin{array}{c}
       \Gamma;\Delta_1\tpjg\exp_1:\vwty_1\otimes\vwty_2 \\
       \Gamma;\Delta_2,x_1:\vwty_1,x_2:\vwty_2\tpjg\exp_2:\vwty \\
       \end{array}$$} \\[6pt]
\infer[\mbox{\bf(ty-lam-l)}]
      {\Gamma;\Delta\tpjg\lam{x}{\exp}:\vwty_1\ltimp\vwty_2}
      {(\Gamma;\Delta),x:\vwty_1\tpjg\exp:\vwty_2} \\[6pt]
\infer[\mbox{\bf(ty-app-l)}]
      {\Gamma;\Delta_1,\Delta_2\tpjg\app{\exp_1}{\exp_2}:\vwty_2}
      {\Gamma;\Delta_1\tpjg\exp_1:\vwty_1\ltimp\vwty_2 &
       \Gamma;\Delta_2\tpjg\exp_2:\vwty_1 } \\[6pt]
\infer[\mbox{\bf(ty-lam-i)}]
      {\Gamma;\emptyset\tpjg\lam{x}{\exp}:\vwty_1\itimp\vwty_2}
      {(\Gamma;\emptyset),x:\vwty_1\tpjg\exp:\vwty_2 &
       \resof(\exp) = \emptyset} \\[6pt]
\infer[\mbox{\bf(ty-app-i)}]
      {\Gamma;\Delta_1,\Delta_2\tpjg\app{\exp_1}{\exp_2}:\vwty_2}
      {\Gamma;\Delta_1\tpjg\exp_1:\vwty_1\itimp\vwty_2 &
       \Gamma;\Delta_2\tpjg\exp_2:\vwty_1 } \\[6pt]
\infer[\mbox{\bf(ty-fix)}]
      {\Gamma;\emptyset\tpjg\fix{f}{\val}:\ty}
      {\Gamma, f:\ty;\emptyset\tpjg\val:\ty} \\[6pt]
\infer[\mbox{\bf(ty-pool)}]
      {\tpjg\pool:\vwty}
      {$$\begin{array}{c}
       (\emptyset;\emptyset)\tpjg\pool(0):\vwty \\
       (\emptyset;\emptyset)\tpjg\pool(\tid):\tunit~~\mbox{for each $0<\tid\in\dom(\pool)$} \\
       \end{array}$$} \\[6pt]
\end{array}\]
\caption{The typing rules for $\langz$}
\label{figure:langz:typing_rules}
\end{figure}
\subsection{Static Semantics}
We present typing rules for $\langz$ in this section.  It is required that
each variable occur at most once in an intuitionistic (linear) expression
context $\Gamma$ ($\Delta$), and thus $\Gamma$ ($\Delta$) can be regarded
as a finite mapping.  Given $\Gamma_1$ and $\Gamma_2$ such that
$\dom(\Gamma_1)\cap\dom(\Gamma_2)=\emptyset$, we write
$(\Gamma_1,\Gamma_2)$ for the union of $\Gamma_1$ and $\Gamma_2$.  The same
notation also applies to linear expression contexts ($\Delta$).  Given an
intuitionistic expression context $\Gamma$ and a linear expression context
$\Delta$, we can form a combined expression context $(\Gamma;\Delta)$ if
$\dom(\Gamma)\cap\dom(\Delta)=\emptyset$.  Given $(\Gamma;\Delta)$, we may
write $(\Gamma;\Delta),x:\vwty$ for either $(\Gamma;\Delta,x:\vwty)$ or
$(\Gamma,x:\vwty;\Delta)$ (if $\vwty$ is actually a type).

We use $\tsubst$ for a substitution on type and viewtype variables:
\[\begin{array}{rcl}
\tsubst & ::= & [] \mid \tsubst[\tvar\mapsto\ty] \mid \tsubst[\vtvar\mapsto\vwty] \\
\end{array}\]
Given a viewtype $\vwty$, we write $\vwty[\tsubst]$ for the result of
applying $\tsubst$ to $\vwty$, which is defined in a standard manner.
Given a constant resource $\cres$, we write $\SIG\temd\cres:\vtbase$ to mean
that $\cres$ is assigned the viewtype $\vtbase$ in the signature $\SIG$.
Given a constant $\const$, we use the following judgment:
\[\begin{array}{c}
\SIG\temd\const:(\vwty^0_1,\ldots,\vwty^0_n)\ctimp\vwty^0 \\
\end{array}\]
to mean that $\const$ is assigned a c-type of the form
$(\vwty_1,\ldots,\vwty_n)\ctimp\vwty$ (in the signature $\SIG$) and there
exists $\tsubst$ such that $\vwty^0_i=\vwty_i[\tsubst]$ for $1\leq i\leq n$
and $\vwty^0=\vwty[\tsubst]$. In other words,
$(\vwty^0_1,\ldots,\vwty^0_n)\ctimp\vwty^0$ is an instance of
$(\vwty_1,\ldots,\vwty_n)\ctimp\vwty$. In the case where $\vwty^0_i$ is a
true viewtype for some $i$, $\vwty^0$ is required to also be a true
viewtype if the c-type is assigned to a constructor (rather than a
function).

A typing judgment in $\langz$ is of the form
$(\Gamma;\Delta)\tpjg\exp:\vwty$, meaning that $\exp$ can be
assigned the viewtype $\vwty$ under $(\Gamma;\Delta)$.  The typing rules
for $\langz$ are listed in Figure~\ref{figure:langz:typing_rules}.

By inspecting the rules in Figure~\ref{figure:langz:typing_rules},
we can readily see that a closed value cannot contain any resources if the
value itself can be assigned a type (rather than a linear type).
More formally, we have the following proposition:
\begin%
{proposition}
\label{proposition:value_type}
Assume that $(\emptyset;\emptyset)\tpjg\val:\ty$ is derivable.
Then $\resof(\val)=\emptyset$.
\end{proposition}
This proposition plays a fundamental role in the design of $\langz$
as the rules in Figure~\ref{figure:langz:typing_rules} are actually
so formulated in order to make it hold.

The following lemma, which is often referred to as
{\em Lemma of Canonical Forms}, relates the form of a value to its type:
\begin{lemma}\label{lemma:langz:canonical}
Assume that $(\emptyset;\emptyset)\tpjg\val:\vwty$ is derivable.
\begin{itemize}
\item
If $\vwty=\tbase$,
then $\val$ is of the form $\ccon(\val_1,\ldots,\val_n)$.
\item
If $\vwty=\vtbase$,
then $\val$ is of the form $\cres$ or $\ccon(\val_1,\ldots,\val_n)$.
\item
If $\vwty=\tunit$, then $\val$ is $\dunit$.
\item
If
$\vwty=\ty_1*\ty_2$
or
$\vwty=\vwty_1\otimes\vwty_2$,
then $\val$ is of the form $\tuple{\val_1,\val_2}$.
\item
If $\vwty=\vwty_1\itimp\vwty_2$ or $\vwty=\vwty_1\ltimp\vwty_2$, then
$\val$ is of the form $\lam{x}{\exp}$.
\end{itemize}
\end{lemma}
\begin{proof}
By an inspection of the rules in Figure~\ref{figure:langz:typing_rules}.
\end{proof}

We use $\esubst$ for substitution on variables $\xf$:
\[\begin{array}{rcl}
\esubst & ::= & [] \mid \esubst[x\mapsto\val] \mid \esubst[f\mapsto e] \\
\end{array}\]
For each $\esubst$, we define the multiset $\resof(\esubst)$ of resources
in $\esubst$ as follows:
$$\resof(\esubst)=\uplus_{\xf\in\dom(\esubst)}\resof(\esubst(\xf))$$
Given an expression $\exp$, we use $\exp[\esubst]$ for the result of
applying $\esubst$ to $\exp$, which is defined in a standard manner.
We write $(\Gamma_1;\Delta_1)\tpjg\esubst:(\Gamma_2;\Delta_2)$ to
mean that
\begin{itemize}
\item
$\dom(\esubst)=\dom(\Gamma_2)\cup\dom(\Delta_2)$, and
\item
$(\Gamma_1;\emptyset)\tpjg\esubst(\xf):\Gamma_2(\xf)$
is derivable for each $\xf\in\Gamma_2$, and
\item
there exists a linear expression context $\Delta_{1,x}$ for each
$x\in\dom(\Delta_2)$ such that
$(\Gamma_1;\Delta_{1,x})\tpjg\esubst(x):\Delta_2(x)$
is derivable, and
\item
$\Delta_1=\cup_{x\in\dom(\Delta_2)}\Delta_{1,x}$
\end{itemize}
The following lemma, which is often referred to as
{\em Substitution Lemma}, is needed to establish the soundness of the type
system of $\langz$:
\begin{lemma}(Substitution)\label{lemma:langz:substitution}
Assume $(\Gamma_1;\Delta_1)\tpjg\esubst:(\Gamma_2;\Delta_2)$ and
$(\Gamma_2;\Delta_2)\tpjg\exp:\vwty$.
Then $(\Gamma_1;\Delta_1)\tpjg\exp[\esubst]:\vwty$ is derivable
and $\resof(\exp[\esubst])=\resof(\exp)\uplus\resof(\esubst)$.
\end{lemma}
\begin{proof}
By induction on the derivation of $(\Gamma_2;\Delta_2)\tpjg\exp:\vwty$.
\end{proof}

\subsection
{Dynamic Semantics}
\label{section:langz:dynamics}
We present evaluation rules for $\langz$ in this section.
The evaluation contexts in $\langz$ are defined below:
\[\begin{array}{lrcl}
\mbox{eval.~ctx.}
& E & ::= & \kern144pt \\
\kern12pt
\hbox to 0pt
{$[] \mid \const(\vval,E,\vexp) \mid \dif(E,\exp_1,\exp_2) \mid$ \hss} \\
\kern12pt
\hbox to 0pt
{$\tuple{E,\exp} \mid \tuple{\val, E} \mid \letin{\tuple{x_1,x_2}=E}{\exp} \mid$} \\
\kern12pt
\hbox to 0pt
{$\dfst{E} \mid \dsnd{E} \mid \app{E}{\exp} \mid \app{\val}{E}$} \\
\end{array}\]
Given an evaluation context $E$ and an expression $\exp$, we use $E[\exp]$
for the expression obtained from replacing the only hole $[]$ in $E$ with
$\exp$.

\def\subst#1#2#3{#3[#2\mapsto #1]}
\begin%
{definition}
We define pure redexes and their reducts as follows.
\begin%
{itemize}
\item
$\dif(\ctrue,\exp_1,\exp_2)$ is a pure redex whose reduct is $\exp_1$.
\item
$\dif(\cfalse,\exp_1,\exp_2)$ is a pure redex whose reduct is $\exp_2$.
\item
$\letin{\tuple{x_1,x_2}=\tuple{\val_1,\val_2}}{\exp}$
is a pure redex whose reduct is $\subst{\val_1,\val_2}{x_1,x_2}{\exp}$.
\item
$\dfst{\tuple{\val_1,\val_2}}$ is a pure redex whose reduct is $\val_1$.
\item
$\dsnd{\tuple{\val_1,\val_2}}$ is a pure redex whose reduct is $\val_2$.
\item
$\app{\lam{x}{\exp}}{\val}$
is a pure redex whose reduct is $\subst{\val}{x}{\exp}$.
\item
$\fix{f}{\val}$ is a pure redex whose reduct is $\subst{\fix{f}{\val}}{f}{\val}$.
\end{itemize}
\end{definition}

\def\eval{\rightarrow}
\def\Eval{\Rightarrow}
\def\frandbit{\mbox{\it randbit}}
Evaluating calls to constant functions is of particular importance in
$\langz$. Assume that $\cfun$ is a constant function of arity $n$. The
expression $\cfun(v_1,\ldots,v_n)$ is an {\em ad-hoc} redex if $\cfun$ is
defined at $v_1,\ldots,v_n$, and any value of $\cfun(v_1,\ldots,v_n)$ is a
reduct of $\cfun(v_1,\ldots,v_n)$. For instance, $1+1$ is an ad hoc redex
and $2$ is its sole reduct. In contrast, $1+\ctrue$ is not a redex as it is
undefined. We can even have non-deterministic constant functions.  For
instance, we may assume that the ad-hoc redex $\frandbit()$ can evaluate to
both 0 and 1.

Let $\exp$ be a well-typed expression of the form
$\cfun(v_1,\ldots,v_n)$ and $\resof(\exp)\subseteq\rns$ holds for some
valid $\rns$ (that is, $\rns\in\RES$).  We always assume that there
exists a reduct $\val$ in $\langz$ for $\cfun(v_1,\ldots,v_n)$ such
that $(\rns\backslash\resof(e))\uplus\resof(v)\in\RES$. By doing so,
we are able to give a presentation with much less clutter.

\begin%
{definition} Given expressions $\exp_1$ and $\exp_2$, we write
$\exp_1\eval\exp_2$ if $\exp_1=E[\exp]$ and $\exp_2=E[\exp']$ for some
$E,\exp$ and $\exp'$ such that $\exp'$ is a reduct of $\exp$, and we may
say that $\exp_1$ evaluates or reduces to $\exp_2$ purely if $\exp$ is a
pure redex.
\end{definition}

Note that resources may be generated as well as consumed when ad-hoc
reductions occur. This is an essential issue of great importance in any
linear type system designed to support practical programming.

\begin%
{definition}
Given pools $\pool_1$ and $\pool_2$, the relation $\pool_1\eval\pool_2$ is
defined according to the following rules:
\end{definition}
\begin{center}
\fontsize{8}{8}\selectfont
\[\begin{array}{c}
\infer[\mbox{(PR0)}]
      {\pool[\tid\mapsto\exp_1]\eval\pool[\tid\mapsto\exp_2]}
      {\exp_1\eval\exp_2} \\[6pt]
\infer[\mbox{(PR1)}]
      {\pool\eval\pool[\tid_0:=E[\tuple{}]][\tid\mapsto\app{\lam{x}{e}}{\tuple{}}]}
      {\pool(\tid_0)=E[\fthreadcreate(\lam{x}{e})]} \\[6pt]
\infer[\mbox{(PR2)}]
      {\pool[\tid\mapsto\tuple{}]\eval\pool}{\tid > 0}
\end{array}\]
\end{center}
If a pool $\pool_1$ evaluates to another pool $\pool_2$ by the rule (PR0),
then one program in $\pool_1$ evaluates to its counterpart in $\pool_2$ and
the rest stay the same; if by the rule (PR1), then a fresh program is
created; if by the rule (PR2), then a program (that is not the main
program) is eliminated.

From this point on, we always (implicitly) assume that
$\resof(\pool)\in\RES$ holds whenever $\pool$ is well-typed.  The soundness
of the type system of $\langz$ rests upon the following two theorems:
\begin%
{theorem}
(Subject Reduction on Pools)
\label{theorem:langz:subject_reduction_on_pools}
Assume that $\tpjg\pool_1:\vwty$ is derivable and $\pool_1\eval\pool_2$
holds for some $\pool_2$ satisfying $\resof(\pool_2)\in\RES$. Then
$\tpjg\pool_2:\vwty$ is also derivable.
\end{theorem}
\begin%
{proof}
By structural induction on the derivation of $\tpjg\pool_1:\vwty$.
Note that Lemma~\ref{lemma:langz:substitution} is needed.
\end{proof}

\begin%
{theorem}
(Progress Property on Pools)
\label{theorem:langz:progress_on_pools}
Assume that $\tpjg\pool_1:\vwty$ is derivable. Then we have the following
possibilities:
\begin{itemize}
\item
$\pool_1$ is a singleton mapping $[0\mapsto\val]$ for some $\val$, or
\item
$\pool_1\eval\pool_2$ holds for some $\pool_2$ such that $\resof(\pool_2)\in\RES$.
\end{itemize}
\end{theorem}
\begin%
{proof}
By structural induction on the derivation of $\tpjg\pool_1:\vwty$.  Note
that Lemma~\ref{lemma:langz:canonical} is needed. Essentially, we can
readily show that $\Pi_1(\tid)$ for any $\tid\in\dom(\Pi_1)$ is either a
value or of the form $E[\exp]$ for some evaluation context $E$ and redex
$\exp$.  If $\Pi_1(\tid)$ is a value for some $\tid>0$, then this value
must be $\tuple{}$.  So the rule $\mbox{(PR2)}$ can be used to reduce
$\Pi_1$.  If $\Pi_1(\tid)$ is of the form $E[e]$ for some redex $e$, then
the rule $\mbox{(PR0)}$ can be used to reduce $\Pi_1$.
\end{proof}

By combining Theorem~\ref{theorem:langz:subject_reduction_on_pools} and
Theorem~\ref{theorem:langz:progress_on_pools}, we immediately conclude that
the evaluation of a well-typed pool either leads to a pool that itself is a
singleton mapping of the form $[0\mapsto\val]$ for some value $\val$, or it
goes on forever.  In other words, $\langz$ is type-sound.

\section%
{Extending $\langz$ with Channels}
\label{section:langch}
There is no support for communication between threads in $\langz$, making
$\langz$ uninteresting as a multi-threaded language. We extend $\langz$ to
$\langch$ with support for synchronous communication channels in this
section. Supporting asynchronous communication channels is certainly
possible but would result in a more involved theoretical development.
We do support both synchronous and asynchronous session-typed communication
channels in practice, though. In order to assign types to channels, we
introduce session types as follows:
\[
\begin%
{array}{rcl}
\ST & ::= & \stnil \mid \chsnd{\vwty}{\ST} \mid \chrcv{\vwty}{\ST} \\
\end{array}
\]
An empty session is specified by $\stnil$.  When used to specify a {\em
  positive} channel, $\chsnd{\vwty}{\ST}$ means to send onto the channel a
value of the viewtype $\vwty$ and $\chrcv{\vwty}{\ST}$ means to receive
from the channel a value of the viewtype $\vwty$.  Dually, when used to
specify a {\em negative} channel, $\chsnd{\vwty}{\ST}$ means to receive
from the channel a value of the viewtype $\vwty$ and $\chrcv{\vwty}{\ST}$
means to send onto the channel a value of the viewtype $\vwty$. After
either sending or receiving is done, the channel is specified by $\ST$.

Formally, the dual of a session type is defined as follows:
\[
\begin%
{array}{rcl}
\dual{\stnil} & = & \stnild \\
\dual{\chsnd{\vwty}{\ST}} & = & \chrcv{\vwty}{\dual{\ST}} \\
\dual{\chrcv{\vwty}{\ST}} & = & \chsnd{\vwty}{\dual{\ST}} \\
\end{array}
\]
where $\stnild$ is another constant session type denoting the dual of
$\stnil$.  Traditionally, $\stnil$ and $\stnild$ are treated as the same
constant in the study on session types. In our implementation, a
positive channel specified by $\stnil$ awaits a message to close itself
while a negative channel specified by $\stnil$ (that is, a positive channel
specified by $\stnild$) sends out such a message before closing itself.

Formally, we use $\tchpos(\ST)$ and $\tchneg(\ST)$ for a positive and
negative channel specified by $\ST$, respectively. Though it is clear that
$\tchpos(\ST)$ and $\tchneg(\dual{\ST})$ equal $\tchneg(\ST)$ and
$\tchpos(\dual{\ST})$, respectively, we do not attempt for now to make use
of this fact in our formalization of session types. In particular, there is
currently no support for turning a positive channel into a negative channel
or vice versa.

We use $\stvar$ as a variable ranging over session types.  The function
$\fchnegcreate$ for creating a negative channel is assigned the following
c-type:
\[
\begin%
{array}{rcl}
\fchnegcreate & : &
(\tchpos(\stvar)\ltimp\tunit)\ctimp\tchneg(\stvar) \\
\end{array}\]
Given a linear function of the type $\tchpos(\ST)\ltimp\tunit$ for some
$\ST$, $\fchnegcreate$ essentially creates a positive channel and a
negative channel that are properly connected, and then starts a thread for
evaluating the call that applies the function to the positive channel, and
then returns the negative channel. The newly created positive channel and
negative channel share the same id.

The send and receive functions for positive channels
are given the following c-types:
\[
\begin%
{array}{rcl}
\fsendpos & : &
(\tchpos(\chsnd{\vtvar}{\stvar}), \vtvar)\ctimp\tchpos(\stvar) \\
\frecvpos & : &
(\tchpos(\chrcv{\vtvar}{\stvar}))\ctimp\tchpos(\stvar)\otimes\vtvar \\
\end{array}\]
Note that $\fsendpos$ and $\frecvpos$ correspond to the functions
$\mbox{\tt chanpos\_send}$ and $\mbox{\tt chanpos\_recv}$, respectively.

Dually,
the send and receive functions for negative channels
are given the following c-types:
\[
\begin%
{array}{rcl}
\frecvneg & : &
(\tchneg(\chrcv{\vtvar}{\stvar}), \vtvar)\ctimp\tchneg(\stvar) \\
\fsendneg & : &
(\tchneg(\chsnd{\vtvar}{\stvar}))\ctimp\tchneg(\stvar)\otimes\vtvar \\
\end{array}\]
Note that $\fsendneg$ and $\frecvneg$ correspond to the functions
$\mbox{\tt channeg\_send}$ and $\mbox{\tt channeg\_recv}$, respectively.

The functions $\fclosepos$ and $\fcloseneg$ for closing positive and
negative channels, respectively, are given the following c-types:
\[
\begin%
{array}{rcl}
\fclosepos~~:~~(\tchpos(\stnil))\ctimp\tunit
&\kern-4pt&
\fcloseneg~~:~~(\tchneg(\stnil))\ctimp\tunit \\
\end{array}\]
Note that $\fclosepos$ and $\fcloseneg$ can also be referred to as
$\mbox{\tt chanpos\_close}$ and $\mbox{\tt channeg\_close}$, respectively.

\def\chpres{\mbox{ch}}
\def\chnres{\overline{\mbox{ch}}}
In $\langch$, there are resource constants $\chpres_i$ and $\chnres_i$ for
positive and negative channels, respectively, where $i$ ranges over natural
numbers. For each $i$, $\chpres_i$ and $\chnres_i$ are dual to each other
and their channel ids are $i$.  We use $\chpcst$ and $\chncst$ to range
over $\chpres_i$ and $\chnres_i$, respectively, referring one as the dual
of the other.

There are no new typing rules in $\langch$ over $\langz$.
Given a session type $\ST$, we say that the type $\tchpos(\ST)$ matches the
type $\tchneg(\ST)$ and vice versa. In any type derivation of $\pool:\vwty$
satisfying $\resof(\pool)\in\RES$, the type assigned to a positive channel
$\chpcst$ is always required to match the one assigned to the corresponding
negative channel $\chncst$ of the same channel id.
For evaluating pools in $\langz$, we have the following additional rules in
$\langch$:
\begin{center}
\fontsize{8}{8}\selectfont
\[\begin{array}{c}
\infer[\mbox{(PR3)}]
      {\pool\eval\pool[\tid_0:=E[\chncst]][\tid\mapsto\app{\lam{x}{e}}{\chpcst}]}
      {\pool(\tid_0)=E[\fchnegcreate(\lam{x}{e})]} \\[6pt]
\infer[\mbox{(PR4-clos)}]
      {\pool\eval\pool[\tid_1:=E_1[\dunit]][\tid_2:=E_2[\dunit]]}
      {\pool(\tid_1)=E_1[\fclosepos(\chpcst)]&&\pool(\tid_2)=E_2[\fcloseneg(\chncst)]} \\[6pt]
\infer[\mbox{(PR4-send)}]
      {\pool\eval\pool[\tid_1:=E_1[\chpcst]][\tid_2:=E_2[(\chncst,\val)]]}
      {\pool(\tid_1)=E_1[\fsendpos(\chpcst, \val)]&&\pool(\tid_2)=E_2[\fsendneg(\chncst)]} \\[6pt]
\infer[\mbox{(PR4-recv)}]
      {\pool\eval\pool[\tid_1:=E_1[\chncst]][\tid_2:=E_2[(\chpcst,\val)]]}
      {\pool(\tid_1)=E_1[\frecvneg(\chncst, \val)]&&\pool(\tid_2)=E_2[\frecvpos(\chpcst)]} \\[6pt]
\end{array}\]
\end{center}
For instance, the rule $\mbox{PR4-send}$ states: If a program in a pool is
of the form $E_1[\fsendpos(\chpcst, \val)]$ and another of the form
$E_2[\fsendneg(\chncst)]$, then this pool can be reduced to another pool by
replacing these two programs with $E_1[\chpcst]$ and $E_2[(\chncst,v)]$,
respectively.

While Theorem~\ref{theorem:langz:subject_reduction_on_pools} (Subject
Reduction) can be readily established for $\langch$,
Theorem~\ref{theorem:langz:progress_on_pools} (Progress) requires some
special treatment due to the presence of session-typed primitive functions
$\fchnegcreate$, $\fclosepos$, $\fcloseneg$, $\fsendpos$, $\frecvpos$,
$\fsendneg$, and $\frecvneg$.

A partial (ad-hoc) redex in $\langch$ is of one of the following forms:
$\fclosepos(\chpcst)$, $\fcloseneg(\chncst)$, $\fsendpos(\chpcst, \val)$,
$\frecvpos(\chpcst)$, $\fsendneg(\chncst)$, and $\frecvneg(\chncst, \val)$.
Clearly, either a positive channel $\chpcst$ or a negative channel
$\chncst$ is involved in each partial redex.  We say that
$\fclosepos(\chpcst)$ and $\fcloseneg(\chncst)$ match, and
$\fsendpos(\chpcst, \val)$ and $\fsendneg(\chncst)$ match, and
$\frecvpos(\chpcst)$ and $\frecvneg(\chncst, \val)$ match.  We can
immediately prove in $\langch$ that each well-typed program is either a
value or of the form $E[\exp]$ for some evaluation context $E$ and
expression $\exp$ that is either a redex or a partial redex. We refer to an
expression as a {\it blocked} one if it is of the form $E[\exp]$ for some
partial redex $\exp$. We say two blocked expressions $E_1[\exp_1]$ and
$E_2[\exp_2]$ match if $\exp_1$ and $\exp_2$ are matching partial
redexes. Clearly, a pool containing two matching blocked expressions can be
reduced according to one of the rules $\mbox{PR4-clos}$, $\mbox{PR4-send}$,
and $\mbox{PR4-recv}$.

Intuitively, a pool $\Pi$ is deadlocked if $\Pi(tid)$ for
$\tid\in\dom(\Pi)$ are all blocked expressions but there are no matching
ones among them, or if $\Pi(0)$ is a value and $\Pi(tid)$ for positive
$\tid\in\dom(\Pi)$ are all blocked expressions but there are no matching
ones among them. The following lemma states that a well-typed pool in
$\langch$ can never be deadlocked:

\begin%
{lemma}
(Deadlock-Freedom)
\label{lemma:langch:deadlock-freeness}
Let $\Pi$ be a well-typed pool in $\langch$ such that $\Pi(0)$ is either a
value containing no channels or a blocked expression and $\Pi(tid)$ for
each positive $\tid\in\dom(\Pi)$ is a blocked expression. If $\Pi$ is
obtained from evaluating an initial pool containing no channels, then there
exist two thread ids $\tid_1$ and $\tid_2$ such that $\Pi(\tid_1)$ and
$\Pi(\tid_2)$ are matching blocked expressions.
\end{lemma}
Note that it is entirely possible to encounter a scenario where the main
thread in a pool returns a value containing a channel while another thread
is waiting for something to be sent on the channel. Technically, we do not
classify this scenario as a deadlocked one. There are many forms of values
that contain channels. For instance, such a value can be a channel itself,
or a closure-function containing a channel in its environment, or a compound
value like a tuple that contains a channel as one part of it, etc. Clearly,
any value containing a channel can only be assigned a true viewtype.

\def\MCH{M}
\def\MCHS{{\cal M}}
The primary technical contribution of the paper lies in the following presented
approach to establishing Lemma~\ref{lemma:langch:deadlock-freeness}. 
Let us use $\MCH$ for sets of (positive and negative) channels and $\MCHS$
for a finite non-empty collection (that is, multiset) of such sets. We say
that $\MCHS$ is {\it regular} if the sets in $\MCHS$ are pairwise disjoint
and each pair of channels $\chpcst$ and $\chncst$ are either both included
in the multiset union $\biguplus(\MCHS)$ of all the sets in $\MCHS$ or both
excluded from it. Of course, $\biguplus(\MCHS)$ is the same as the set
union $\bigcup(\MCHS)$ as the sets in $\MCHS$ are pairwise disjoint.

\def\DFred{\leadsto}
Let $\MCHS$ be a regular collection of channel sets. We say that $\MCHS$
{\it DF-reduces} to $\MCHS'$ via $\chpcst$ if there exist $\MCH_1$ and
$\MCH_2$ in $\MCHS$ such that $\chpcst\in\MCH_1$ and $\chncst\in\MCH_2$ and
$\MCHS'=(\MCHS\backslash\{\MCH_1,\MCH_2\})\cup\{\MCH_{12}\}$, where
$\MCH_{12}=(\MCH_1\cup\MCH_2)\backslash\{\chpcst,\chncst\}$.  We say that
$\MCHS$ DF-reduces to $\MCHS'$ if $\MCHS$ DF-reduces to $\MCHS'$ via some
$\chpcst$. We may write $\MCHS\DFred\MCHS'$ to mean that $\MCHS$
DF-reduces to $\MCHS$. We say that $\MCHS$ is DF-normal if there is no $\MCHS'$
such that $\MCHS\DFred\MCHS'$ holds.

\begin%
{proposition}
Let $\MCHS$ be a regular collection of channel sets.  If $\MCHS$ is
DF-normal, then each set in $\MCHS$ consists of an indefinite number of
channel pairs $\chpcst$ and $\chncst$. In other words, for each $\MCH$
in a DF-normal $\MCHS$, a channel $\chpcst$ is in $\MCH$ if and only if
its dual $\chncst$ is also in $\MCH$.
\end{proposition}
\begin
{proof}
The proposition immediately follows from the definition of DF-reduction
$\leadsto$.
\end{proof}

\begin%
{definition}
\label{def:DFreducibility}
A regular collection $\MCHS$ of channel sets is DF-reducible if either (1)
each set in $\MCHS$ is empty or (2) $\MCHS$ is not DF-normal and $\MCHS'$
is DF-reducible whenever $\MCHS\DFred\MCHS'$ holds.
\end{definition}
We say that a channel set $\MCH$ is self-looping if it contains both
$\chpcst$ and $\chncst$ for some $\chpcst$.  Obviously, a regular collection
$\MCHS$ of channel sets is not DF-reducible if there is a self-looping $\MCH$
in $\MCHS$.

\begin%
{proposition}
\label{prop:DFreducibility:0}
Let $\MCHS$ be a regular collection of channel sets.  If
$\MCHS$ is DF-reducible and $\MCHS'=\MCHS\backslash\{\emptyset\}$,
then $\MCHS'$ is also DF-reducible.
\end{proposition}
\begin
{proof}
Straightforwardly.
\end{proof}

\begin%
{proposition}
\label{prop:DFreducibility:1}
Let $\MCHS$ be a regular collection of channel sets.  If
$\MCHS\DFred\MCHS'$ and $\MCHS'$ is DF-reducible, then $\MCHS$ is also
DF-reducible.
\end{proposition}
\begin
{proof}
Clearly, $\MCHS\DFred\MCHS'$ via some $\chpcst$.  Assume
$\MCHS\DFred\MCHS_1$ via $\chpcst_1$ for some $\MCHS_1$ and $\chpcst_1$.
If $\chpcst$ and $\chpcst_1$ are the same, then $\MCHS_1$ is DF-reducible
as it is the same as $\MCHS'$. Otherwise, it can be readily verified that
there exists $\MCHS'_1$ such that $\MCHS_1\DFred\MCHS'_1$ via $\chpcst$ and
$\MCHS'\DFred\MCHS'_1$ via $\chpcst_1$. Clearly, the latter implies
$\MCHS'_1$ being DF-reducible.  Note that the size of $\MCHS_1$ is strictly
less than that of $\MCHS$.  By induction hypothesis on $\MCHS_1$, we have
$\MCHS_1$ being DF-reducible.  By definition, $\MCHS$ is DF-reducible.
\end{proof}

\begin%
{proposition}
\label{prop:DFreducibility:2}
Let $\MCHS$ be a regular collection of channel sets that is DF-reducible.
If $\MCH_1$ and $\MCH_2$ in $\MCHS$ contain $\chpcst$ and $\chncst$,
respectively, then
$\MCHS'=(\MCHS\backslash\{\MCH_1,\MCH_2\})\cup\{\MCH'_1,\MCH'_2\}$ is also
DF-reducible, where $\MCH'_1=\MCH_1\backslash\{\chpcst\}$ and
$\MCH'_2=\MCH_2\backslash\{\chncst\}$.
\end{proposition}
\begin%
{proof}
The proposition follows from a straightforward induction on the size of
the set union $\bigcup(\MCHS)$.
\end{proof}
\begin%
{lemma}
\label{lemma:DFreducibility:3}
Let $\MCHS$ be a regular collection of $n$ channel sets
$\MCH_1,\ldots,\MCH_n$ for some $n\geq 1$.  If the union
$\bigcup(\MCHS)=\MCH_1\cup\ldots\cup\MCH_n$ contains at least $n$ channel
pairs $(\chpcst_1,\chncst_1),\ldots,(\chpcst_n,\chncst_n)$, then $\MCHS$ is
not DF-reducible.
\end{lemma}
\begin%
{proof}
By induction on $n$.  If $n=1$, then $\MCHS$ is not DF-reducible as
$\MCH_1$ is self-looping.  Assume $n > 1$.  If either $\MCH_1$ or $\MCH_2$
is self-looping, then $\MCHS$ is not DF-reducible.  Otherwise, we may
assume that $\chpcst_1\in\MCH_1$ and $\chncst_1\in\MCH_2$ without loss of
generality.  Then $\MCHS$ DF-reduces to $\MCHS'$ via $\chpcst_1$ for some
$\MCHS'$ containing $n-1$ channel sets. Note that $\bigcup(\MCHS')$
contains at least $n-1$ channel pairs
$(\chpcst_2,\chncst_2),\ldots,(\chpcst_n,\chncst_n)$. By induction
hypothesis, $\MCHS'$ is not DF-reducible. So $\MCHS$ is not DF-reducible,
either.
\end{proof}
\def\fMCH{\rho_{\it CH}}
\def\fMCHS{{\cal R}_{\it CH}}
Given an expression $\exp$ in $\langch$, we use
$\fMCH(\exp)$ for the set of channels contained in $\exp$.
Given a pool $\Pi$ in $\langch$, we use $\fMCHS(\Pi)$ for the
collection of $\fMCH(\Pi(tid))$, where $\tid$ ranges over $\dom(\Pi)$.
\begin%
{lemma}
\label{lemma:DFreducibility:4}
If $\fMCHS(\Pi)$ is DF-reducible and $\Pi$ evaluates to $\Pi'$, then
$\fMCHS(\Pi')$ is also DF-reducible.
\end{lemma}
\begin%
{proof}
Note that $\fMCHS(\Pi)$ and $\fMCHS(\Pi')$ are the same unless
$\Pi$ evaluates
to $\Pi'$ according to one of the rules $\mbox{PR3}$, $\mbox{PR4-clos}$,
$\mbox{PR4-send}$, and $\mbox{PR4-recv}$.
\begin%
{itemize}
\item
For the rule $\mbox{PR3}$: We have $\fMCHS(\Pi')\DFred\fMCHS(\Pi)$ via the
newly introduced channel $\chpcst$. By
Proposition~\ref{prop:DFreducibility:1}, $\fMCHS(\Pi')$ is DF-reducible.
\item
For the rule $\mbox{PR4-clos}$: We have that $\fMCHS(\Pi')$ is DF-reducible
by Proposition~\ref{prop:DFreducibility:2}.
\item
For the rule $\mbox{PR4-send}$: Let $\chpcst$ be the channel on which a
value is sent when $\Pi$ evaluates to $\Pi'$. Note that this value can
itself be a channel or contain a channel. We have $\fMCHS(\Pi)\DFred\MCHS$
via $\chpcst$ for some $\MCHS$. So $\MCHS$ is DF-reducible by definition.
Clearly, $\fMCHS(\Pi')\DFred\MCHS$ via $\chpcst$ as well.  By
Proposition~\ref{prop:DFreducibility:2}, $\fMCHS(\Pi')$ is DF-reducible.
\item For the rule $\mbox{PR4-recv}$: This case is similar to the previous one.
\end{itemize}
\end{proof}
We are now ready to give a proof for
Lemma\ref{lemma:langch:deadlock-freeness}:
\vspace{6pt}

\noindent%
{\it Proof}~~
Note that any channel, either positive or negative, can appear at most once
in $\fMCHS(\Pi)$, and a channel $\chpcst$ appears in $\fMCHS(\Pi)$ if and
only if its dual $\chncst$ also appears in $\fMCHS(\Pi)$. In addition, any
positive channel $\chpcst$ being assigned a type of the form $\tchpos(S)$
in the type derivation of $\Pi$ for some session type $S$ mandates that its
dual $\chncst$ be assigned the type of the form $\tchneg(S)$.

Assume that $\Pi(\tid)$ is a blocked expression for each
$\tid\in\dom(\Pi)$. If the partial redex in $\Pi(\tid_1)$ involves a
positive channel $\chpcst$ while the partial redex in $\Pi(\tid_2)$
involves its dual $\chncst$, then these two partial redexes must
match. This is due to $\Pi$ being well-typed. In other words, the ids
of the channels involved in the partial redexes of $\Pi(\tid)$ for
$\tid\in\dom(\Pi)$ are all distinct. This simply implies that there
are $n$ channel pairs $(\chpcst,\chncst)$ in $\bigcup(\fMCHS(\Pi))$
for some $n$ greater than or equal to the size of $\Pi$.  By
Lemma~\ref{lemma:DFreducibility:3}, $\fMCHS(\Pi)$ is not reducible.
On the other hand, $\fMCHS(\Pi)$ is reducible by
Lemma~\ref{lemma:DFreducibility:4} as $\Pi_0$ evaluates to $\Pi$ (in
many step) and $\fMCHS(\Pi_0)$ (containing only sets that are empty)
is reducible. This contradiction indicates that there exist $\tid_1$
and $\tid_2$ such that $\Pi(\tid_1)$ and $\Pi(\tid_2)$ are matching
blocked expressions. Therefore $\Pi$ evaluates to $\Pi'$ for some pool
$\Pi'$ according to one of the rules $\mbox{PR4-clos}$,
$\mbox{PR4-send}$, and $\mbox{PR4-recv}$.

With Proposition~\ref{prop:DFreducibility:0}, the case can be handled
similarly where $\Pi(0)$ is a value containing no channels and $\Pi(\tid)$
is a blocked expression for each positive
$\tid\in\dom(\Pi)$.~\hfill\fillsquare

\vspace{6pt}
Please assume for the moment that we would like to add into $\langch$ a function
$\fchnegcreatetwo$ of the following type:
\[
\begin%
{array}{c}
((\tchpos(\stvar_1),\tchpos(\stvar_2))\ltimp\tunit)\ctimp(\tchneg(\stvar_1),\tchneg(\stvar_2)) \\
\end{array}\]
One may think of $\fchnegcreatetwo$ as a slight generalization of
$\fchnegcreate$ that creates in a single call two channels instead of one.
Unfortunately, adding $\fchnegcreatetwo$ into $\langch$ can potentially
cause a deadlock. For instance, we can easily imagine a scenario where the
first of the two channels $(\chncst_1,\chncst_2)$ returned from a call to
$\fchnegcreatetwo$ is used to send the second to the newly created thread
by the call, making it possible for that thread to cause a deadlock by
waiting for a value to be sent on $\chpcst_2$.  Clearly,
Lemma~\ref{lemma:DFreducibility:4} is invalidated if $\fchnegcreatetwo$ is
added.

\vspace{6pt}
\noindent%
The soundness of the type system of $\langch$ rests upon the following
two theorems
(corresponding to
Theorem~\ref{theorem:langz:subject_reduction_on_pools}
and Theorem~\ref{theorem:langz:progress_on_pools}):
\begin%
{theorem}
(Subject Reduction on Pools)
\label{theorem:langch:subject_reduction_on_pools}
Assume that $\tpjg\pool_1:\vwty$ is derivable and
$\pool_1\eval\pool_2$ such that $\resof(\pool_2)\in\RES$.
Then $\tpjg\pool_2:\vwty$ is derivable.
\end{theorem}
\begin%
{proof}
The proof is essentially the same as the one for
Theorem~\ref{theorem:langz:subject_reduction_on_pools}.  The only
additional part is for checking that the rules $\mbox{PR3}$,
$\mbox{PR4-clos}$, $\mbox{PR4-send}$, and $\mbox{PR4-recv}$ are all
consistent with respect to the typing rules listed in
Figure~\ref{figure:langz:typing_rules}.
\end{proof}

\begin%
{theorem}
(Progress Property on Pools)
\label{theorem:langch:progress_on_pools}
Assume that $\tpjg\pool_1:\vwty$ is derivable and
$\resof(\pool_1)$ is valid.
Also assume that $\resof(v)$ contains no channels
for every value $v$ of the type $\vwty$. Then we have
the following possibilities:
\begin{itemize}
\item
$\pool_1$ is a singleton mapping $[0\mapsto\val]$ for some $\val$, or
\item
$\pool_1\eval\pool_2$ holds for some $\pool_2$ such that $\resof(\pool_2)\in\RES$.
\end{itemize}
\end{theorem}
\begin%
{proof}
The proof follows the same structure as the one for
Theorem~\ref{theorem:langz:progress_on_pools}.
Lemma~\ref{lemma:langch:deadlock-freeness} is needed to handle the
case where all of the threads (possibly excluding the main thread) in a
pool consist of blocked expressions.
\end{proof}

\section%
{Additional Features for $\langch$}
\label{section:additional}
We briefly mention certain additional features for $\langch$
that are to be used in some examples presented later.
\label{section:langch2}
\subsection{Bidirectional Forwarding}
There is a special primitive function of the name $\fchposneglink$ for
connecting a positive channel with a negative channel specified by the
following session type:
\[
\begin%
{array}{rcl}
\fchposneglink & : &
(\tchpos(\stvar), \tchneg(\stvar))\ctimp\tunit \\
\end{array}\]
Given a positive channel and a negative channel, $\fchposneglink$
sends each value received from the positive channel onto the negative
channel and vice versa. In other words, $\fchposneglink$ does
bidirectional forwarding between these two channels. In practice,
$\fchposneglink$ is often used to implement delegation of service.  It
can be readily verified that the two channels passed to a call to
$\fchposneglink$ can never have the same channel id; if one is
$\chpcst_1$, then the other must be $\chncst_2$ for some $\chpcst_2$
distinct from $\chpcst_1$. Calling $\fchposneglink$ on a positive
channel and its dual surely results in a deadlock. One of the
evaluation rules for $\fchposneglink$ is given as follows:
\begin{center}
\fontsize{8}{8}\selectfont
\[\begin{array}{c}
\infer
      {\pool\eval\pool[\tid_1:=E_1[\chncst_1]][\tid_2:=E_2[\exp_2]][\tid_3:=E_3[(\chpcst_2,\val)]]}
      {\pool(\tid_1)=E_1[\exp_1]&&\pool(\tid_2)=E_2[\exp_2]&&\pool(\tid_3)=E_3[\exp_3]} \\[6pt]
\end{array}\]
\end{center}
where we have $\exp_1=\frecvneg(\chncst_1, \val)$,
$\exp_2=\fchposneglink(\chpcst_1,\chncst_2)$, and
$\exp_3=\frecvpos(\chpcst_2)$. The other ones are omitted.  It should be
clear that Lemma~\ref{lemma:DFreducibility:4} still holds after
$\fchposneglink$ is added, and thus
Lemma~\ref{lemma:langch:deadlock-freeness} still holds as well.

\subsection{User-Defined Datatypes}
The kind of (recursive) datatypes in ML (for tagged unions) can be added
into $\langch$ without any difficulty. In terms of theory, it is
straightforward to support user-defined (recursive) session datatypes in
$\langch$, allowing sessions of indefinite length to be
specified. Essentially, all we need is to add folding/unfolding rules for
handling recursive session types. As for implementation, we currently
support recursive session types based on an indirect approach, which is
illustrated through the examples presented in
Section~\ref{section:examples}.

\subsection{Quantification over Types}
We can readily incorporate both universally and existentially quantified
types into $\langch$. In ATS, there are predicative quantification and
impredicative quantification. The former is for dependent types (of
DML-style~\cite{XP99,DML-jfp07}) while the latter for parametric
polymorphism. For instance, the example presented in
Section~\ref{subsection:examples:queue} makes use of both parametric
polymorphism and DML-style dependent types.  In terms of theory, we can
readily incorporate quantified session types into $\langch$. As for
implementation, we have not yet attempted to add into ATS direct support
for programming with quantified session types. Instead, we rely on an
indirect approach to do so, which is illustrated in
Section~\ref{section:LLinterp}.

\def\tservice{{\bf service}}
\section%
{Interpreting Linear Logic Connectives}
\label{section:LLinterp}
\def\limplies{\multimap}
Unlike logic-based formalizations of session
types~\cite{ToninhoCP12,Wadler12}, we have not introduced session type
constructors that are directly based on or related to logic connectives in
linear logic. In this section, we interpret in $\langch$ some common linear
logic connectives including multiplicative conjunction ($\otimes$),
multiplicative implication ($\limplies$), additive disjunction ($\&$), and
additive conjunction ($\oplus$). We are unclear as to whether
multiplicative disjunction ($\invamp$) can be handled at all. We also
briefly mention the exponential connective $!$ at the end.

Note that the presented code is written in the concrete syntax of ATS,
which is largely ML-like. We expect that people who can read ML code should
have no great difficulty in following the presented ATS code as it makes
only use of common functional programming features.\\[-6pt]

\begin%
{figure}
\fontsize{8pt}{9pt}\selectfont
\begin%
{verbatim}
fun
fserv_times{A,B:type}
(
  chp: chpos(chsnd(chneg(A))::B)
, chn_a: chneg(A), chn_b: chneg(B)
) : void = let
//
val () =
  chanpos_send (chp, chn_a)
//
in
  chanposneg_link (chp, chn_b)
end // end of [fserv_times]
\end{verbatim}
\caption{Interpreting multiplicative conjunction}
\label{figure:fserv_times}
\end{figure}

\begin%
{figure}
\fontsize{8pt}{9pt}\selectfont
\begin%
{verbatim}
fun
fserv_implies{A,B:type}
(
  chp: chpos(chrcv(chneg(A))::B)
, fchn: chneg(A) -<lincloptr1> chneg(B)
) : void = let
  val chn_b = fchn(chanpos_recv(chp))
in
  cloptr_free(fchn); chanposneg_link(chp, chn_b)
end // end of [fserv_implies]
\end{verbatim}
\caption{Interpreting (multiplicative) implication}
\label{figure:fserv_implies}
\end{figure}
\noindent{\bf M-Conjunction($\otimes$):}\kern6pt
Given two session types $A$ and $B$, a channel of the session type
$A\otimes B$ can be interpreted as one that inputs a channel specified
by $A$ and then behaves as a channel specified by
$B$~\cite{ToninhoCP12,Wadler12}. This interpretation is from the
client's viewpoint, meaning in $\langch$ that $A\otimes B$ should be
defined as $\chsnd{\tchneg(A)}{B}$. Clearly, any reasonable
interpretation for $A\otimes B$ is expected to allow the construction
of a channel of the type $\tchneg(A\otimes B)$ based on two channels
of the types $\tchneg(A)$ and $\tchneg(B)$ and vice versa. In
Figure~\ref{figure:fserv_times}, a function \verb|fserv_times| is
implemented to allow a channel of the type $\tchneg(A\otimes B)$ to be
built by the following call:
\begin
{center}
\fontsize{8pt}{9pt}\selectfont
\verb|channeg_create(llam(chp) => fserv_times(chp,chn_a,chn_b))|
\end{center}
where \verb|chn_a| and \verb|chn_b| are channels of the types
$\tchneg(A)$ and $\tchneg(B)$, respectively, and the keyword
\verb|llam| forms a linear function (that is to be called exactly
once). The other direction (that is, obtaining $\tchneg(A)$ and
$\tchneg(B)$ from $\tchneg(A\otimes B)$) is straightforward and thus
skipped.

Of course, one can also interpret $A\otimes B$ as
$\chsnd{\tchneg(A)}{\chsnd{\tchneg(B)}{\stnil}}$. With this interpretation,
it should be obvious to see how $\tchneg(A\otimes B)$ can be constructed based
on $\tchneg(A)$ and $\tchneg(B)$ and vice versa.\\[-6pt]

\noindent{\bf M-Implication($\limplies$):}\kern6pt
Given two session types $A$ and $B$, a channel of the session type
$A\limplies B$ can be interpreted as one that outputs a channel specified
by $A$ and then behaves as a channel specified by
$B$~\cite{ToninhoCP12,Wadler12}. This interpretation is from the client's
viewpoint, meaning in $\langch$ that $A\limplies B$ should be defined as
$\chrcv{\tchneg(A)}{B}$. If one has a function \verb|fchn| that turns a
negative channel specified by A into a negative channel specified by B,
then one can build as follows a negative channel specified by $A\limplies
B$:
\begin
{center}
\fontsize{8pt}{9pt}\selectfont
\verb|channeg_create(llam(chp) => fserv_limplies(chp, fchn))|
\end{center}
where the function \verb|fserv_implies| is implemented in
Figure~\ref{figure:fserv_implies}.  Note that the function
\verb|cloptr_free| is called to explicitly free a linear function that has
been called.\\[-6pt]

\begin%
{figure}[t!]
\fontsize{8pt}{9pt}\selectfont
\begin%
{verbatim}
datatype
channeg_adisj
  (A:type, B:type, type) =
| channeg_adisj_l(A, B, A) of () // tag=0
| channeg_adisj_r(A, B, B) of () // tag=1

extern
fun
channeg_adisj{A,B:type}
(
  chn: !chneg(adisj(A, B)) >> chneg(X)
) : #[X:type] channeg_adisj(A, B, X)

extern
fun
chanpos_adisj_l{A,B:type}
  (chp: !chpos(adisj(A,B)) >> chpos(A)): void
and
chanpos_adisj_r{A,B:type}
  (chp: !chpos(adisj(A,B)) >> chpos(B)): void

(* ****** ****** *)

fun
fserv_adisj_l{A,B:type}
(
  chp: chpos(adisj(A, B)), chn: chneg(A)
) : void = let
  val()=chanpos_adisj_l(chp) in chanposneg_link(chp, chn)
end // end of [fserv_adisj_l]

fun
fserv_adisj_r{A,B:type}
(
  chp: chpos(adisj(A, B)), chn: chneg(B)
) : void = let
  val()=chanpos_adisj_r(chp) in chanposneg_link(chp, chn)
end // end of [fserv_adisj_r]
\end{verbatim}
\caption{Interpreting additive disjunction}
\label{figure:fserv_adisj_lr}
\end{figure}
\noindent{\bf A-Disjunction($\oplus$)}\kern6pt
Given two session types $A$ and $B$, a channel of the session type
$A\oplus B$ can be interpreted as one that inputs a boolean value
and then behaves as a channel specified by either $A$ or $B$ depending
on whether the boolean value is true or false, respectively.
This interpretation is from the client's viewpoint,
meaning in $\langch$ that $A\oplus B$ should be defined as follows:
\begin%
{center}
$\forall b:bool.~\chsnd{\tbool(b)}{\tchoose(A, B, b)}$
\end{center}
where $\tbool(b)$ is a singleton type for the boolean value
equal to $b$, and $\tchoose(A, B, b)$ equals either $A$ or $B$
depending on whether $b$ equals true or false, respectively.
It should be noted that the positive channel type
$\tchpos(\forall b:bool.~\chsnd{\ldots}{\ldots})$
should be interpreted as follows:
\begin%
{center}
$\forall b:bool.~\tchpos(\chsnd{\tbool(b)}{\tchoose(A, B, b)})$
\end{center}
but the negative one
$\tchneg(\forall b:bool.~\chsnd{\ldots}{\ldots})$ is as follows:
\begin%
{center}
$\exists b:bool.~\tchneg(\chsnd{\tbool(b)}{\tchoose(A, B, b)})$
\end{center}
We have not yet added into ATS direct support for quantified session types
like the ones appearing here. Instead, we rely on an indirect approach to
handle such types.  In Figure~\ref{figure:fserv_adisj_lr}, the type
constructor \verb|adisj| stands for $\oplus$.  We declare a datatype
\verb|channeg_adisj| so as to introduce two tags: \verb|channeg_adisj_l|
and \verb|channeg_adisj_r| are internally represented as 0 and 1,
respectively.  When applied to a negative channel specified by some session
type \verb|adisj(A,B)|, \verb|channeg_adisj| returns a tag received from the
channel. Note that the syntax \verb|[X:type]| represents an existential
quantifier and the symbol \verb|#| in front of it means that the type
variable \verb|X| to the left of the quantifier is also in its
scope. Essentially, the type assigned to \verb|channeg_adisj| indicates
that the tag returned from a call to the function determines the type of
the argument of the function after the call: 0 means \verb|chneg(A)| and 1
means \verb|chneg(B)| here. Applied to a positive channel,
\verb|chanpos_adisj_l| and \verb|chanpos_adisj_r| send 0 and 1 onto the
channel, respectively. By now, it should be clearly that the following call
turns a channel \verb|chn| of the type \verb|chneg(A)| into one of the type
\verb|chneg(adisj(A,B))|:
\begin
{center}
\fontsize{8pt}{9pt}\selectfont
\verb|channeg_create(llam(chp) => fserv_adisj_l(chp, chn))|
\end{center}
Similarly,
the following call turns a channel \verb|chn| of the type
\verb|chneg(B)| into one of the type \verb|chneg(adisj(A,B))|:
\begin
{center}
\fontsize{8pt}{9pt}\selectfont
\verb|channeg_create(llam(chp) => fserv_adisj_r(chp, chn))|
\end{center}
Often, the kind of choice associated with $A\oplus B$ is referred to
as {\it internal} choice as it is the server that determines whether
A or B is chosen.\\

\begin%
{figure}[t!]
\fontsize{8pt}{9pt}\selectfont
\begin%
{verbatim}
datatype
chanpos_aconj
  (A:type, B:type, type) =
| chanpos_aconj_l(A, B, A) of () // tag=0
| chanpos_aconj_r(A, B, B) of () // tag=1

extern
fun
chanpos_aconj{A,B:type}
(
  !chpos(aconj(A, B)) >> chpos(X)
) : #[X:type] chanpos_aconj(A, B, X)

extern
fun
channeg_aconj_l{A,B:type}
  (!chneg(aconj(A,B)) >> chneg(A)): void
and
channeg_aconj_r{A,B:type}
  (!chneg(aconj(A,B)) >> chneg(B)): void

(* ****** ****** *)

datavtype
choose(a:vtype, b:vtype, bool) =
  | choose_l (a, b, true) of (a)
  | choose_r (a, b, false) of (b)

fun
fserv_aconj
  {A,B:type}
(
  chp: chpos(aconj(A, B))
, fchn:
  {b:bool}
    bool(b) -<lincloptr1> choose(chneg(A), chneg(B), b)
  // fchn
) : void = let
  val opt = chanpos_aconj(chp)
in
//
case opt of
| chanpos_aconj_l() => let
    val~choose_l(chn) = fchn(true)
  in
    chanposneg_link(chp, chn)
  end // end of [chanpos_aconj_l]
| chanpos_aconj_r() => let
    val~choose_r(chn) = fchn(false)
  in
    chanposneg_link(chp, chn)
  end // end of [chanpos_aconj_r]
//
end // end of [fserv_aconj]
\end{verbatim}
\caption{Interpreting additive conjunction}
\label{figure:fserv_aconj}
\end{figure}
\noindent{\bf A-Conjunction($\&$)}\kern6pt
Given two session types $A$ and $B$, a channel of the session type $A\& B$
can be interpreted as one that outputs a boolean value and then behaves as
a channel specified by either $A$ or $B$ depending on whether the boolean
value is true or false, respectively.  This interpretation is from the
client's viewpoint, meaning in $\langch$ that $A\& B$ should be defined as
follows:
\begin%
{center}
$\forall b:bool.~\chrcv{\tbool(b)}{\tchoose(A, B, b)}$
\end{center}
Note that the positive channel type
$\tchpos(\forall b:bool.~\chrcv{\ldots}{\ldots})$
should be interpreted as follows in this case:
\begin%
{center}
$\exists b:bool.~\tchpos(\chrcv{\tbool(b)}{\tchoose(A, B, b)})$
\end{center}
but the negative one
$\tchneg(\forall b:bool.~\chrcv{\ldots}{\ldots})$ is as follows:
\begin%
{center}
$\forall b:bool.~\tchneg(\chrcv{\tbool(b)}{\tchoose(A, B, b)})$
\end{center}
The code in Figure~\ref{figure:fserv_aconj} essentially shows how a channel
of the type $\tchneg(A\& B)$ can be built based on a value of the type
$\tchneg(A)\& \tchneg(B)$. Note that \verb|aconj| refers to $\&$ when
called to form a session type. As $\&$ is not supported as a type
constructor in ATS, we use the following one to represent it:
\begin%
{center}
$\forall b:bool.~\tbool(b) \ltimp \tchoose(\tchneg(A), \tchneg(B), b)$
\end{center}
The declared datatype \verb|choose| is a linear one.  The symbol
\verb|~| in front of a linear data constructor like \verb|choose_l| means
that the constructor itself is freed after the arguments of the constructor
are taken out. What \verb|fserv_conj| does is clear: It checks the tag
received on its first argument \verb|chp| (a positive channel) and then
determine whether to offer \verb|chp| as a channel specified by A or B. So
the kind of choice provided by $\&$ is external: It is the client that
decides whether A or B is chosen.\\[-6pt]

\def\fservicecreate{\mbox{\it service\_create}}
\noindent{\bf Exponential($!$)}\kern6pt
Given a session type $\ST$, we have a type $\tservice(\ST)$ that can be
assigned to a value representing a {\it persistent} service specified by
$\ST$.  With such a service, channels of the type $\tchneg(\ST)$ can be
created repeatedly. A built-in function $\fservicecreate$ is assigned the
following type for creating a service:
\[\begin%
{array}{rcl}
\fservicecreate & : &
(\tchpos(\stvar)\itimp\tunit)\ctimp\tservice(\stvar) \\
\end{array}\]
In contrast with $\fchnegcreate$ for creating a channel, $\fservicecreate$
requires that its argument be a non-linear function (so that this function
can be called repeatedly).

\begin%
{figure}[t!]
\fontsize{8pt}{9pt}\selectfont
\begin%
{verbatim}
datatype
chanpos_list
  (a:type, type) =
| chanpos_list_nil(a, nil) of () // tag=0
| chanpos_list_cons(a, snd(a)::sslist(a)) of () // tag=1

extern
fun
chanpos_list{a:type}
(
  chp: !chanpos(sslist(a)) >> chanpos(ss)
) : #[ss:type] chanpos_list(a, ss)

extern
fun
channeg_list_nil{a:type}
  (chn: !chneg(sslist(a)) >> chneg(nil)): void
and
channeg_list_cons{a:type}
(
  chn: !chneg(sslist(a)) >> chneg(snd(a)::sslist(a))
) : void // end-of-function
\end{verbatim}
\caption{Functions for unfolding session list channels}
\label{figure:sieve:sslist-unfolding}
\end{figure}

\begin%
{figure}[t!]
\fontsize{8pt}{9pt}\selectfont
\begin%
{verbatim}
implement
ints_filter
  (chn, n0) = let
//
fun
getfst
(
  chn: !chneg(sslist(int))
) : int = let
//
val () =
  channeg_list_cons(chn)
//
val fst = channeg_send(chn)
//
in
  if fst mod n0 > 0 then fst else getfst(chn)
end // end of [getfst]
//
fun
fserv
(
  chp: chpos(sslist(int))
, chn: chneg(sslist(int))
) : void = let
//
val opt = chanpos_list(chp)
//
in
//
case opt of
| chanpos_list_nil() =>
  (
    chanpos_close(chp);
    channeg_list_close(chn)
  )
| chanpos_list_cons() =>
  (
    chanpos_send(chp, getfst(chn)); fserv(chp, chn)
  )
//
end // end of [fserv]
//
in
  channeg_create(llam(chp) => fserv(chp, chn))
end // end of [ints_filter]
\end{verbatim}
\caption{Ints\_filter: removing multiples of a given integer}
\label{figure:sieve:ints_filter}
\end{figure}

\begin%
{figure}[t!]
\fontsize{8pt}{9pt}\selectfont
\begin%
{verbatim}
implement
sieve() = let
//
fun
fserv
(
  chp: chpos(sslist(int))
, chn: chneg(sslist(int))
) : void = let
//
val opt = chanpos_list(chp)
//
in
//
case opt of
| chanpos_list_nil() =>
  (
    chanpos_close(chp);
    channeg_list_close(chn)
  )
| chanpos_list_cons() => let
    val () =
      channeg_list_cons(chn)
    // end of [val]
    val p0 = channeg_send(chn)
    val chn = ints_filter(chn, p0)
  in
    chanpos_send(chp, p0); fserv(chp, chn)
  end // end of [channeg_list_cons]
//
end // end of [fserv]
//
in
  channeg_create(llam(chp) => fserv(chp, ints_from(2)))
end // end of [sieve]
\end{verbatim}
\caption{Implementing Eratosthenes's sieve}
\label{figure:sieve:sieve}
\end{figure}

\begin%
{figure}[t!]
\fontsize{8pt}{9pt}\selectfont
\begin%
{verbatim}
datatype
chanpos_ssque
  (a:type, type, int) =
| chanpos_ssque_nil(a, nil, 0) of ()
| {n:pos}
  chanpos_ssque_deq(a,snd(a)::ssque(a,n-1),n) of ()
| {n:nat}
  chanpos_ssque_enq(a,rcv(a)::ssque(a,n+1),n) of ()

extern
fun
chanpos_ssque
  {a:type}{n:nat}
  (!chpos(ssque(a,n)) >> chpos(ss))
: #[ss:type] chanpos_ssque(a, ss, n)

extern
fun
channeg_ssque_nil
  {a:type}
  (!chneg(ssque(a,0)) >> chneg(nil)): void
extern
fun
channeg_ssque_deq
  {a:type}{n:pos}
(
  !chneg(ssque(a,n)) >> chneg(snd(a)::ssque(a,n-1))
) : void // end-of-function
and
channeg_ssque_enq
  {a:type}{n:nat}
(
  !chneg(ssque(a,n)) >> chneg(rcv(a)::ssque(a,n+1))
) : void // end-of-function
\end{verbatim}
\caption{Functions for unfolding session queue channels}
\label{figure:queue:ssque-unfolding}
\end{figure}

\begin%
{figure}[t!]
\fontsize{8pt}{9pt}\selectfont
\begin%
{verbatim}
implement
queue_create
  {a}((*void*)) = let
//
fun
fserv
(
  chp: chpos(ssque(a,0))
) : void = let
  val opt = chanpos_ssque(chp)
in
//
case opt of
| chanpos_ssque_nil() =>
    chanpos_close(chp)
| chanpos_ssque_enq() => let
    val x0 = chanpos_recv{a}(chp)
  in
    fserv2(x0, queue_create(), chp)
  end // end of [chanpos_ssque_enq]
//
end // end of [fserv]
//
and
fserv2
{n:nat}
(
  x0: a
, chn: chneg(ssque(a,n))
, chp: chpos(ssque(a,n+1))
) : void = let
  val opt = chanpos_ssque(chp)
in
//
case opt of
| chanpos_ssque_deq() => let
    val () = chanpos_send(chp, x0)
  in
    chanposneg_link(chp, chn)
  end // end of [chanpos_ssque_deq]
| chanpos_ssque_enq() => let
    val y = chanpos_recv{a}(chp)
    val () = channeg_ssque_enq(chn)
    val () = channeg_recv(chn, y)
  in
    fserv2(x0, chn, chp)
  end // end of [chanpos_ssque_enq]
//
end // end of [fserv2]
//
in
  channeg_create(llam(chp) => fserv(chp))
end // end of [queue_create]
\end{verbatim}
\caption{Queue\_create: creating an empty queue of channels}
\label{figure:queue:queue_create}
\end{figure}

\section%
{Examples of Session-Typed Programs}
\label{section:examples}
We present some simple running examples in this section to further
illustrate programming with session types. More practical examples
(e.g. FTP and reversed FTP) are available but difficult to present
here.

\subsection{Eratosthenes's Sieve}
\label{subsection:examples:sieve}
We give an implementation of {\it Eratosthenes's sieve} based on
session-typed channels. In particular, we make use of channels specified
by session lists that are of indefinite length.
This example is essentially taken from SILL~\cite{SILL}.
In Figure~\ref{figure:sieve:sslist-unfolding}, we give three functions for
unfolding channels specified by session lists: one for positive channels
and two for negative channels. Strictly speaking, we should use the name
co-lists (instead of lists) here as it is the client that decides whether
the next element of a list should be generated or not.

For someone familiar with stream-based lazy evaluation, the code in
Figure~\ref{figure:sieve:ints_filter} and Figure~\ref{figure:sieve:sieve}
should be easily accessible.  Given a channel \verb|chn| of integers and an
integer \verb|n0|, the function \verb|ints_filter| builds a new channel of
integers that outputs, when requested, the first integer from \verb|chn|
that is not a multiple of \verb|n0|. Given a channel \verb|chn| of
integers, the function \verb|sieve| outputs, when requested, the first
integer \verb|p0| from \verb|chn|, and then applies \verb|ints_filter| to
\verb|chn| and \verb|p0| to build a new channel, and then applies itself to
the new channel recursively. Note that the code for \verb|ints_from| is
omitted, which returns a channel of all the integers starting from a given
one.

\subsection{A Queue of Channels}
\label{subsection:examples:queue}
We give a queue implementation based on a queue of session-typed channels.
This example is largely based on one in SILL. What is novel here mainly
involves the use of DML-style dependent types to specify the size of each
queue in the implementation.

Given a type \verb|a| and an integer \verb|n|,
we use \verb|ssque(a,n)| as a session type to specify a channel
representing a queue of size \verb|n| in which each element is of the type
\verb|a|. In Figure~\ref{figure:queue:ssque-unfolding}, there are
four functions for unfolding channels specified by session queues: one for
positive channels and three for negative channels.

The function \verb|queue_create| creates a negative channel representing an
empty queue. Based on the code for \verb|queue_create|, we can see that a
queue of size $n$ is represented by a queue of $n+1$ channels (where the
last one always represents an empty queue); each element in the queue is
held in the corresponding channel (or more precisely, the thread running to
support the channel); enqueuing an element is done by sending the element
down to the last channel (representing an empty queue), causing this
channel to create another channel (representing an empty queue); dequeuing
is done by sending out the element held in the first channel (in the queue
of channels) while the thread running to support the channel turns into one
that does bidirectional forwarding.

Note that the implementation of \verb|queue_create| never needs to handle
dequeuing an empty queue or closing a non-empty queues as these operations
are ill-typed, reaping typical benefits from (DML-style) dependent types.

\section%
{Implementing Session-Typed Channels}
\label{section:implementation}
As far as implementation is of the concern, there is very little that needs
to be done regarding typechecking in order to support session types in
ATS. Essentially, the entire effort focuses on implementing session-typed
channels.

\subsection{Implementation in ATS}

The session-typed channels as presented in this paper are first
implemented in ATS.  The parties communicating to each other in a
(dyadic) session run as pthreads.  Each channel is represented as a
record containing two buffers and some locking mechanism (i.e.,
mutexes and conditional variables); a positive channel and its
negative dual share their buffers; the read buffer of a channel is the
write buffer of its dual and vice versa. This implementation (of
session-typed channels) is primarily done for the purpose of obtaining
a proof of concept.

\subsection{Implementation in Erlang}

Another implementation of session-typed channels is done in Erlang.
As the ML-like core of ATS can already be compiled into Erlang, we
have now an option to construct distributed programs in ATS that may
make use of session types and then translate these programs into
Erlang code for execution, thus taking great advantage of the
infrastructural support for distributed computing in Erlang. Each
channel is implemented as a pair of processes; a positive channel
shares with its dual the two processes: One handles read for the
positive channel and write for its negative dual, and the other does
the opposite.  Implementing the functions $\mbox{\tt
  chanpos\_send}$/$\mbox{\tt channeg\_recv}$ and $\mbox{\tt
  chanpos\_recv}$/$\mbox{\tt channeg\_send}$ is straightforward. A
significant complication occurs in our implementation of $\mbox{\tt
  chanposneg\_link}$, which requires the sender of a message to
deliver it at its final destination (instead of having it forwarded
there explicitly).  For a straightforward but much less efficient
implementation of $\mbox{\tt chanposneg\_link}$, one can just rely on
explicit forwarding.

\section{Related Work and Conclusion}
\label{section:related-work-conclusion}

Session types were introduced by Honda~\cite{Honda93} and further
extended~\cite{TakeuchiHK94,HondaVK98}.  There have since been
extensive theoretical studies on session types in the
literature(e.g.,~\cite{HondaYC08,CastagnaDGP09,GayV10,CairesP10,ToninhoCP11,Vasconcelos12,Wadler12}).
However, there is currently rather limited support for practical
programming with session types, and more evidence is clearly needed to
show convincingly that session types can actually be employed
cost-effectively in the construction of relatively large and complex
programs. It is in this context that we see it both interesting and
relevant to study implementation of session types formally.

There are reported implementations of session types in
Java\cite{HuYH08,HuKPYH10,FrancoV13} and other languages
(e.g. Python).  However, these implementations are of a very different
nature when compared to $\langch$. For instance, they mostly focus on
session-typed functionalities being implemented rather than the
(formal) correctness of the implementation of these functionalities.

There are also several implementations of session types in Haskell
(e.g.,~\cite{NeubauerT04,PucellaT08}), which primarily focus on
obtaining certain typeful encodings for session-typed channels.  While
the obtained encodings are shown to be type-correct, there is no
provided mechanism to establish any form of global progress for them.
As is explained in the case of $\fchnegcreatetwo$, one can readily
introduce potential deadlocks inadvertently without breaking
type-correctness.

It is in general a challenging issue to establish deadlock-freeness for
session-typed concurrency. There are variations of session types that
introduce a partial order on time stamps~\cite{SumiiK98} or a
constraint on dependency graphs \cite{abs-1010-5566}.  As for
formulations of session types (e.g.,~\cite{CairesP10,Wadler12}) based
on linear logic~\cite{GIRARD87}, the standard technique for
cut-elimination is often employed to establish global progress (which
implies deadlock-freeness). In $\langch$, there is no explicit tracking
of cut-rule applications in the type derivation of a program (and it
is unclear how such tracking can be done). In essence, the notion of
DF-reducibility (Definition~\ref{def:DFreducibility}) is introduced in
order to carry out cut-elimination in the absence of explicit tracking
of cut-rule applications.

Probably, $\langch$ is most closely related to
SILL~\cite{ToninhoCP13}, a functional programming language that adopts
via a contextual monad a computational interpretation of linear
sequent calculus as session-typed processes.  Unlike in $\langch$, the
support for linear types in SILL is not direct and only monadic values
(representing open process expressions) in SILL can be linear. For
instance, one can readily construct linear data containing channels in
ATS but this is not allowed in SILL. In terms of theoretical
development, the approach to establishing global progress in SILL
cannot be applied to $\langch$ directly. Instead, we see
Lemma~\ref{lemma:DFreducibility:3} as a generalization of the
argument presented in the proof of the theorem on global progress in
SILL. Please see Theorem~5.2~\cite{ToninhoCP13} for details.

Also, $\langch$ is related to previous work on incorporating session
types into a multi-threaded functional language~\cite{VasconcelosRG04},
where a type safety theorem is established to ensure that the
evaluation of a well-typed program can never lead to a so-called
{\em faulty configuration}. However, this theorem does not imply
global progress as a program that is not of faulty configuration can
still deadlock.

As for future work, we are particularly interested in extending $\langch$
with multi-party session types~\cite{HondaYC08}. It will be very interesting
to see whether the approach we use to establish global progress for
$\langch$ can be adapted to handling multi-party session types. We are also
interested in studying session types in restricted settings. For instance,
the number of channels allowed in a program is required to be bounded;
channels (of certain types) may not be sent from one party to another;
etc.

There are a variety of programming issues that need to be addressed in
order to facilitate the use of session types in practice.  Currently,
session types are represented as datatypes in ATS, and programming
with session types often involves writing boilerplate code (e.g., the
code implementing functions like \verb|chanpos_list|,
\verb|channeg_list_nil| and \verb|channeg_list_cons|). In the presence
of large and complex session types, writing such code can tedious and
error-prone. Naturally, we are interested in developing some
meta-programming support for generating such code automatically.
Also, we are in process of designing and implementing session
combinators (in the spirit of parsing combinators~\cite{Hutton92})
that can be conveniently called to assemble subsessions into a
coherent whole.

\end{document}